\documentclass[reprint, amsmath,amssymb,aps,prl,floatfix,superscriptaddress]{revtex4-2}
\usepackage{natbib}
\usepackage{graphicx}
\usepackage{dcolumn}
\usepackage{bm}
\usepackage{dcolumn}
\usepackage{float}

\begin{document}
\title{Linear Classification of Neural Manifolds with Correlated Variability} 
\author{Albert J. Wakhloo}
\affiliation{Center for Computational Neuroscience, Flatiron Institute} \affiliation{Department of Child and Adolescent Psychiatry, New York State Psychiatric Institute}
\author{Tamara J. Sussman}
\affiliation{Department of Child and Adolescent Psychiatry, New York State Psychiatric Institute}
\affiliation{Columbia University Irving Medical College}
\author{SueYeon Chung}
\affiliation{Center for Computational Neuroscience, Flatiron Institute}
\affiliation{Center for Neural Science, New York University}

\begin{abstract}
     Understanding how the statistical and geometric properties of neural activity relate to performance is a key problem in theoretical neuroscience and deep learning. Here, we calculate how correlations between object representations affect the capacity, a measure of linear separability. We show that for spherical object manifolds, introducing correlations between centroids effectively pushes the spheres closer together, while introducing correlations between the axes effectively shrinks their radii, revealing a duality between correlations and geometry with respect to the problem of classification. We then apply our results to accurately estimate the capacity of deep network data.
\end{abstract}

\maketitle

\emph{Introduction:}
Neural networks can learn rich representations of the world. This capacity for representation learning is thought to underlie deep learning's unprecedented success across a wide variety of tasks. However, it is unclear how the geometric and statistical properties of neural network representations shape network performance on common tasks. Recent work addresses this gap by studying the interaction between artificial neural network representations and performance on classification and memorization tasks \cite{rotondo_counting_2020, battista_capacity-resolution_2020, goldt_modeling_2020, farrell2022capacity, biswas_geometric_2022, susman_quality_2021, NEURIPS2019_cfcce062, dahmen_capacity_2020, steinberg2022associative, cohen_soft-margin_2022}, with complementary work in neuroscience studying the interaction between the structure of biological neural network representations and animal behavior \cite{chaudhuri_intrinsic_2019,bernardi_geometry_2020,  chung_neural_2021, sorscher_neural_2022}. Specifically, in \cite{chung_classification_2018, cohen_separability_2020, chung_linear_2016}, the authors introduce the manifold shattering capacity, a measure capturing how easy it is to separate random binary partitions of a set of manifolds with a hyperplane, and express it in terms of the underlying manifold geometry. In this way, network performance on a classification task, as measured by the capacity, can be understood through the geometric structure of the network representations.

\begin{figure}
    \includegraphics[width=8.6cm]{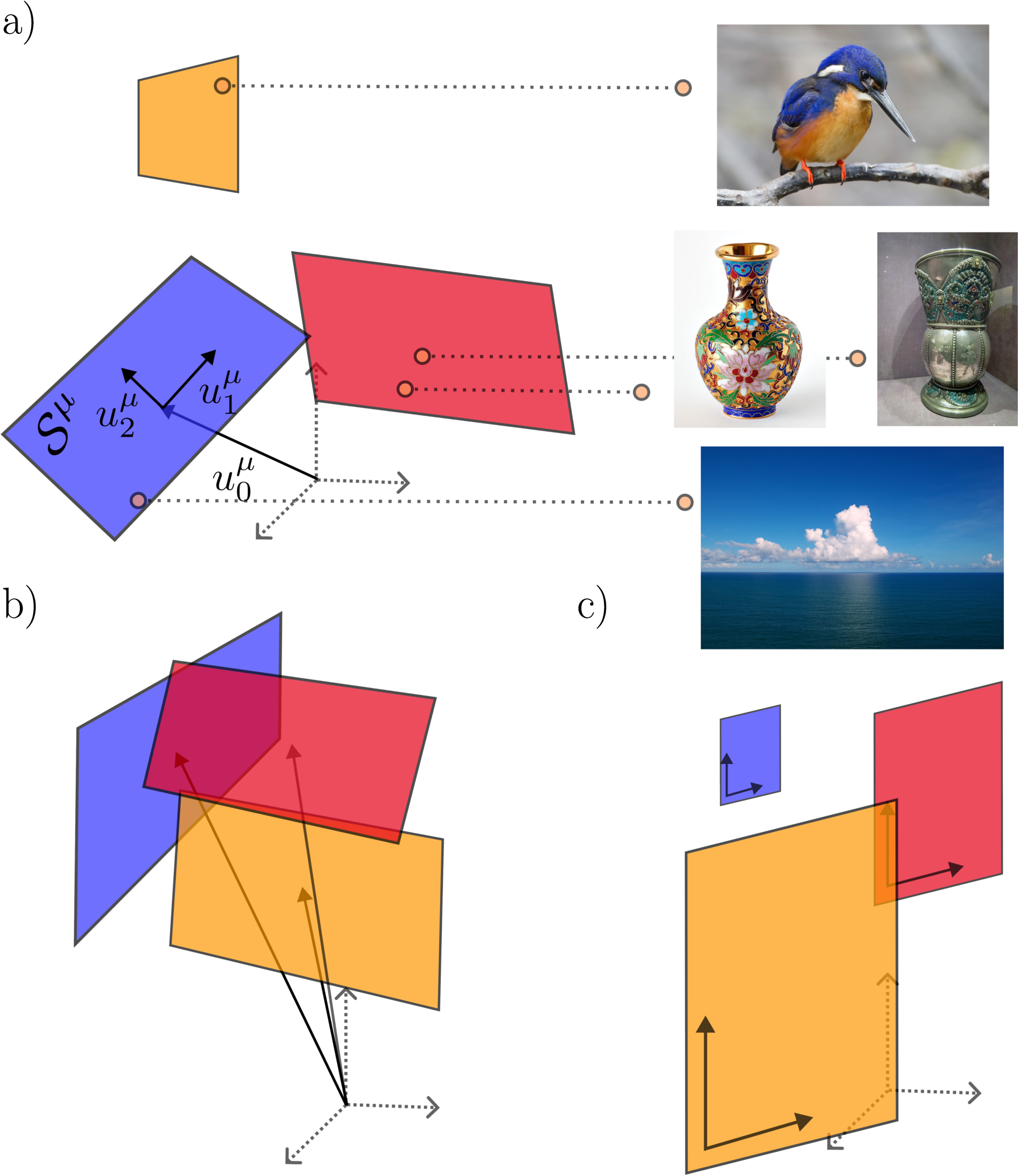}
    \caption{\label{fig:manifold-demo} 
    (a) Neural manifolds arising from different instances of $P=3$ object classes (bird, vase, and cloud \cite{birdimg, vaseimg, chinesevaseimg, skyimg}), with $N=3$ neurons. We parametrize the manifolds in terms of a centroid $u^\mu_0$, axes $u^\mu_{i>0}$, and shape vectors $\mathcal S^\mu$, determining which linear combinations of the axes lie within the manifold. (b) Neural manifolds with correlations in their centroids. (c) Neural manifolds with fully correlated axes. In all three images, different colors correspond to different object class manifolds.}
\end{figure}

Previous works on the manifold capacity have either ignored or coarsely approximated the effects of neural correlations. The best approximation to these effects was reported in \cite{cohen_separability_2020}, where the authors “project out" low-rank correlation structures in manifold centroids. However, the authors find that this approach breaks down when applied to certain artificial network data. Moreover, this approach does not offer analytical insight into the role of different types of correlations in object classification.

Object representations in artificial and biological neural networks exhibit intricate correlation structures, which reflect important properties of the underlying representations \cite{panzeri_structures_2022, zylberberg_robust_2017, morcos2018insights, kornblith2019similarity}. Moreover, as the deep learning community shifts to a self-supervised learning paradigm, many popular loss functions directly enforce particular correlation structures between the latent representations of (possibly augmented) batches of data points \cite{chen2020simple, bardes2022vicreg, zbontar2021barlow, he_momentum_2020}. These considerations call for a theoretical characterization of the relationship between network performance, representational geometry, and the correlation structure of network representations.

In this Letter, we calculate the effects of correlation structures on the capacity. Our formula for the capacity of correlated manifolds generalizes the results in \cite{chung_classification_2018} by stretching the Euclidean norm appearing in previous results in the directions of the eigenvectors of the covariance tensor. We analyze this formula in a simple setting, showing how geometry and correlations interact to determine the capacity, and we go on to apply this formula to accurately estimate the capacity of deep network data.

\emph{Problem statement:} 
Consider a set of $ P $ manifolds, $ M^\mu, $ residing in $ \mathbb R^N $. These manifolds correspond to distinct sets of neuronal activation vectors when presented with different types of stimuli---for example, the set of neural activations for a set of $ P $ classes across all possible class instances in a given layer of an image recognition network [Fig.\;\ref{fig:manifold-demo}(a)]. In what follows, we assume that each manifold resides in an affine subspace of maximal dimension $ K < N.$ That is, for any $ x \in M^\mu $, we have that $ x = u_0^\mu + \sum_{i=1}^K s_i^\mu u^\mu_i $, where $ u_0^\mu $ is a manifold center, $ u_i^\mu $ for $ 1\leq i \leq K $ is a set of manifold axes, and $ s\in \mathcal S^\mu $ are the coordinates of $ x $ with respect to the manifold axes. We use $ \mathcal S^\mu \subset \mathbb R^K $ to denote the set of all possible coordinates in this basis. 

We take the manifold center $u_0^\mu$ to be the average activation of the network layer when presented with a data point from class $ \mu $. The spread of the manifold along the axes therefore corresponds to the network variability as we sample different stimuli from class $ \mu $. Intuitively, manifolds with large centroid norms far away from one another with small spreads along their axes will be easier to classify than large manifolds tightly packed together. 

We now turn to the problem of determining the maximal number of manifolds per dimension, $ \alpha \equiv P/N $, which are, given some random binary labelings $ y^\mu \in \{-1, 1\} $ and some underlying distribution on the $ u^\mu_i $, linearly separable with high probability at a fixed margin $ \kappa$. In what follows, we will be specifically interested in the thermodynamic limit, $N, P \to\infty$ with $P/N=O(1).$ In other words, we find the greatest $ \alpha $ such that there exists a hyperplane with normal $ w \in \mathbb R^N, \;  ||w||^2_2 = N$  satisfying $ \min_{x\in M^\mu}y^\mu \langle w, x \rangle \geq \kappa $ for each manifold $ M^\mu $ with probability 1 in this limit. We define the manifold capacity to be this maximal value of $\alpha$, so that larger capacities imply a more favorable representational geometry for the purpose of classification.  

Following \cite{gardner_space_1988, chung_classification_2018, battista_capacity-resolution_2020, rubin2010theory, schonsberg_efficiency_2021, monasson_properties_1992, Lopez_1995}, we study this problem by calculating the average log-volume of the space of solutions in the thermodynamic limit: 

\begin{gather}
\overline{\log \mathrm{Vol}}=  \overline{\log \int_{\mathbb{S}(\sqrt N)} d^Nw \prod_\mu 
\Theta\bigg( \min_{x \in M^\mu}y^\mu \langle w, x \rangle
- \kappa \bigg)} \; ,
\label{eq:logvol}
\end{gather}

\noindent where $ \mathbb{S}(\sqrt N) $ is the sphere of radius $ \sqrt N $, $\Theta(\cdot)$ is the Heaviside step function, and the average is taken with respect to the quenched disorder in the labels $ y^\mu $ and the axes and centroids $ u_i^\mu $. Viewing the volume as a partition function, we can see that $-N^{-1}\log\mathrm{Vol}$ corresponds to a free energy density, which we assume is self-averaging \cite{mezard_spin_1986}. Given a fixed set of manifold shapes $\mathcal S^\mu $, and choosing the axes and centroids to be independent from one another with $ u_i^\mu \sim \mathcal N(0, N^{-1} I^{(N)}) $, the capacity for such randomly oriented manifolds, $\alpha_{M},$ is given by \cite{chung_classification_2018} 

\begin{gather}
\frac 1 {\alpha_{M}(\kappa)} = \frac 1 P \int D_I T \ \min_{V\in \mathcal A}\sum_{i,\mu} (V_i^\mu - T_i^\mu)^2
\; , 
\label{eq:mf-T}
\end{gather}

\noindent where $D_I T = \prod_{\mu,i} dT^{\mu}_i \exp[-\frac 1 2 (T^\mu_i)^2] / \sqrt{2\pi}$ is an isotropic Gaussian measure and $\mathcal A$ is a convex set of matrices which depends on the geometry of the manifolds, as reflected by their shapes, $\mathcal{S}^\mu$: 

\begin{gather}
\mathcal A \equiv \bigg\{ V \in \mathbb R^{P\times (K+1)}: V^\mu_0 +  \min_{s^\mu \in \mathcal S^\mu} \sum_{i=1}^K V^\mu_i s^\mu_i \geq \kappa \bigg\} .
\label{eq:defA}
\end{gather}

\noindent Note the similarity to the constraint in the $ \Theta $ function in Eq.\;\eqref{eq:logvol} . Indeed, the variable $ V^\mu_i $ corresponds to the inner product of the solution vector $ w $ with the $ i $th axis (or centroid) of the $ \mu$th manifold, multiplied by the label: $V^\mu_i \equiv y^\mu \langle w, u^\mu_i \rangle$. These are the so-called signed fields of the solution vector on the $u^\mu_i$ \cite{chung_classification_2018}. In this way, the capacity can be understood as a function of the geometry of the manifolds as reflected in the set $\mathcal S^\mu.$ In the special case that the manifolds are simply randomly oriented points, the capacity is given by \cite{gardner_space_1988}

\begin{gather}
\frac 1 {\alpha_{point}(\kappa)} = \int_{-\infty}^\kappa \frac{d\xi}{\sqrt{2\pi}} e^{-\frac 1 2 \xi^2}(\xi-\kappa)^2 \; .
\label{eq:cap-pts}
\end{gather}

\noindent From this formula, we can see that the shape sets $\mathcal S^\mu$ cause a lower capacity when compared to that of points.

\emph{Replica theory for correlated manifolds:} 
Here, we consider the situation where manifold axes and centroids are correlated with one another. Intuitively, this corresponds to the fact that different classes in a dataset may be more or less similar to one another in the neural representation space. We enforce correlated axes and centroids by assuming that $ \overline{\langle u^\mu_i,  u^\nu_j \rangle} =C^{\mu,i}_{\nu,j} $ for some positive definite covariance tensor $ C^{\mu,i}_{\nu,j} $. This is done by placing a Gaussian distribution on the centroids and axes: $p(u) \propto \exp\big[- \frac N 2 \sum_{\mu,\nu,i,j,l} (C^{-1})^{\mu,i}_{\nu,j} u^\mu_{i,l} u^\nu_{j,l}\big] $.

We calculate the capacity for correlated manifolds using the replica method \cite{mezard_spin_1986, mezard_information_2009}; the details can be found in the Supplementary Material (SM) \cite{note:SM}. We find that the capacity at a margin $ \kappa $, denoted by $ \alpha_{cor}(\kappa) $, is

\begin{gather}
\frac 1 {\alpha_{cor}(\kappa)} = \frac 1 P \overline{\int D_{y,C}T \min_{V\in \mathcal A}
||V - T||^2_{y, C}}\;,
\label{eq:general-T}
\end{gather}

\noindent where $ D_{y,C} T $ is the zero-mean Gaussian measure with covariance tensor $ y^\mu y^\nu C^{\mu,i}_{\nu,j} $, and the overline denotes the remaining average with respect to the labels $ y^\mu $. Note too that we have defined the Mahalanobis norm: $||X||^2_{y,C} \equiv \sum_{\mu,\nu,i,j} y^\mu y^\nu \big(C^{-1}\big)^{\mu,i}_{\nu,j} X^{\mu}_i X^\nu_j,$ which effectively stretches the Frobenius norm along the eigenvectors of the tensor $y^\mu y^\nu C^{\mu,i}_{\nu,j}$ (Fig.\;\ref{fig:comp-fig}).

\begin{figure}
\includegraphics[width=8.6cm]{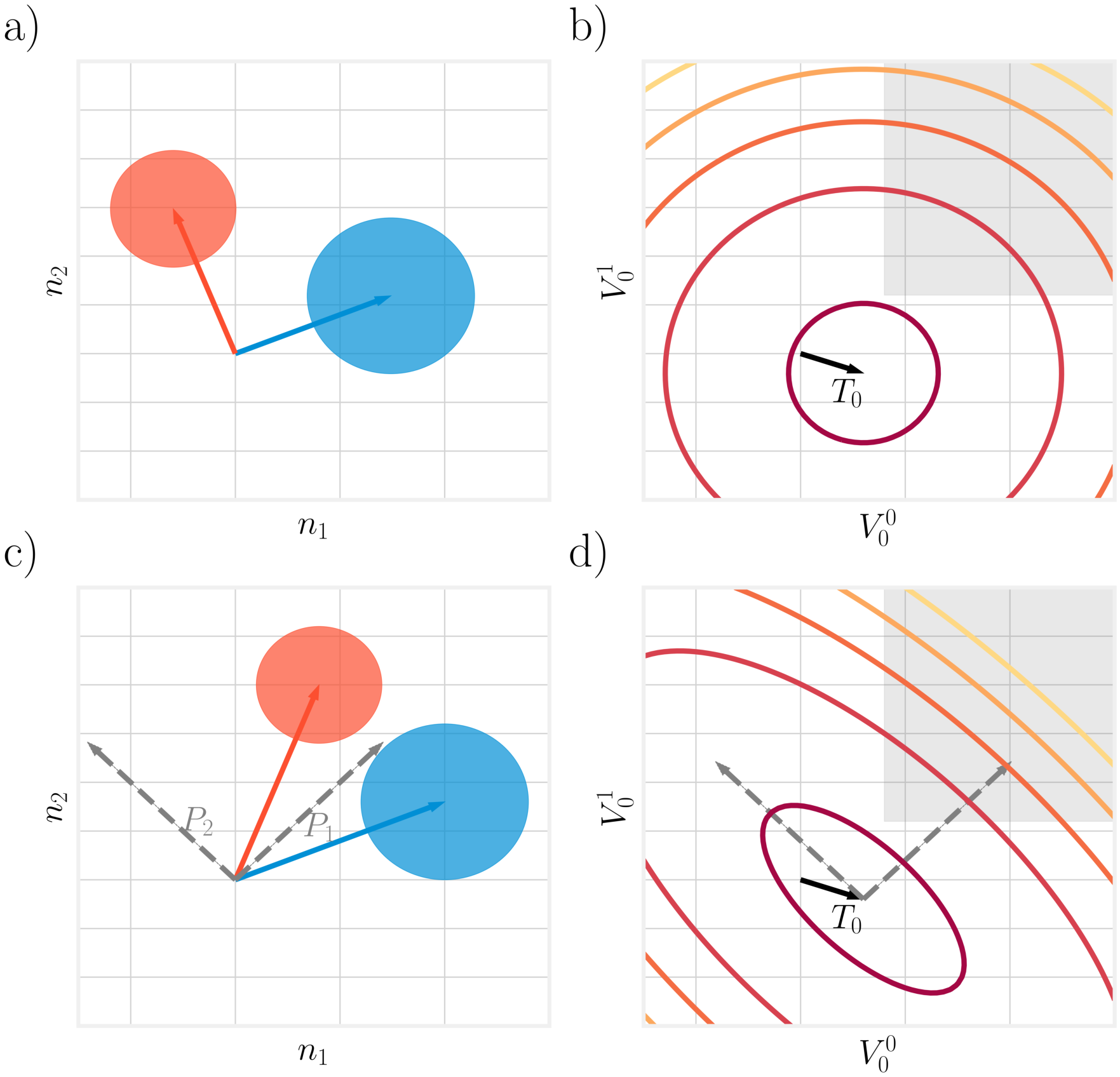}
\caption{\label{fig:comp-fig} The effect of correlations on the optimization landscape for $V_0$. \emph{First column:} Two manifolds with (a) uncorrelated and (c) correlated centroids arising from the activations of two neurons, $n_1$ and $n_2$. \emph{Second column:} Level curves for $||V-T||^2_{y,C}$, given fixed $y$ and $ V_{i>0}$ for the (b) uncorrelated and (d) correlated manifolds. Shaded regions correspond to areas where the constraint is satisfied---i.e., sections of the set $\mathcal{A}$ in Eq.\;\eqref{eq:defA}. Clearly, correlations warp the optimization landscape along the eigenvectors $P_{1}, P_2$ of the centroid covariance matrix with off-diagonal sign flips $y^\mu y^\nu C^{\mu, 0}_{\nu, 0}$.}
\end{figure} 

\emph{Comparison with other capacity estimators:}
It is worth pausing and comparing Eq.\;\eqref{eq:general-T}  to the solution for uncorrelated manifolds in Eq.\;\eqref{eq:mf-T} reported in \cite{chung_classification_2018, cohen_separability_2020}. From Eqs.\;\eqref{eq:mf-T} and \eqref{eq:general-T}, we can see that axes and centroid correlations distort the norm in the minimization from the Euclidean norm to a random Mahalanobis norm which depends on the covariance tensor $ C $ and the random labels $y^\mu$ (Fig.\;\ref{fig:comp-fig}). As such, we expect that the quality of the $\alpha_M$ estimator from Eq.\;\eqref{eq:mf-T} degrades as the manifold axes and centroids become more correlated with one another. We find that this is the case for both $\alpha_M$ and the low-rank approximation method reported in \cite{cohen_separability_2020} when applied to Gaussian point cloud manifolds (Fig.\;\ref{fig:cloud-comparison}). Therefore, the correlated capacity estimator, $\alpha_{cor}$, whose numerical implementation we describe in the SM \cite{note:SM}, should be used whenever working with manifolds with strong correlations (see \cite{note:repository}). 

\emph{The special case of spheres:}
We now look for an answer to the problem we were originally interested in: What are the effects of manifold correlations on the capacity? We answer this question by analytically solving Eq.\;\eqref{eq:general-T}  in a simple setting: $K$-dimensional spheres with homogeneous axis and centroid correlations. More precisely, we assume that the manifold shape sets $ \mathcal{S}^\mu $ are spheres of radius $ 1 $, and the covariance tensor $ C $ is defined by 

\begin{gather}
\label{eq:covar}
C^{\mu,i}_{\nu,j} \equiv 
\begin{cases}
\delta_{i,j}[(1-\lambda) \delta_{\mu,\nu} + \lambda] 
& \text{for } i, j >0
\\
(1 - \psi)\delta_{\mu,\nu} + \psi 
& \text{for } i, j = 0 
\\
 0 & \text{for } i>0, j=0 \; ,
\end{cases}
\end{gather}

\begin{figure}
\includegraphics[width=8.6cm]{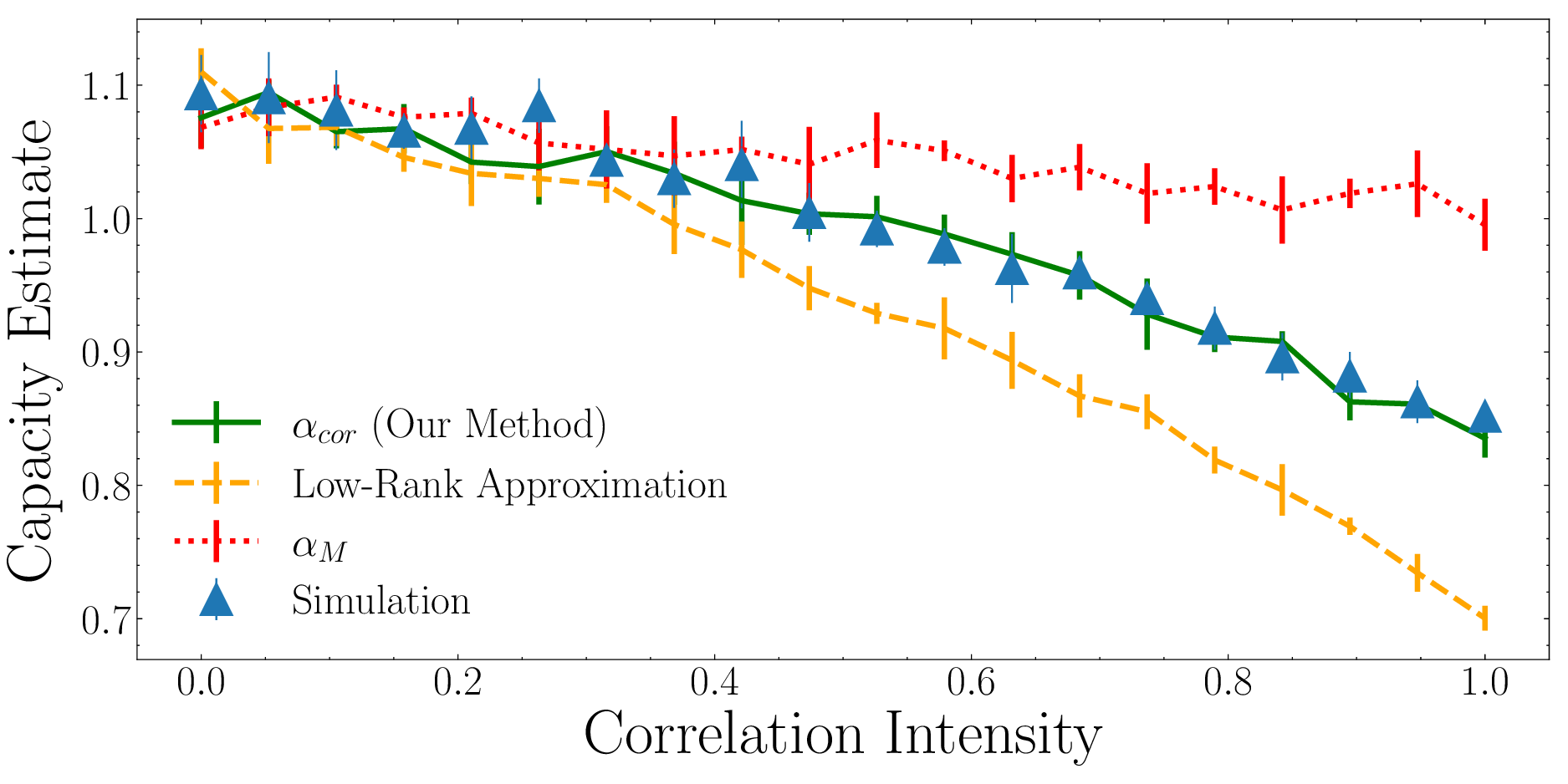}
\caption{\label{fig:cloud-comparison} 
Comparison of three different capacity estimators, including the low-rank approximation of \cite{cohen_separability_2020}, to the numerically estimated ground truth simulation capacity (blue triangles) described in \cite{chung_learning_2018}. The correlation intensity denotes the magnitude of the off-diagonal correlations---see SM for more details \cite{note:SM}.  
}
\end{figure} 

\noindent where $0 \leq \psi, \lambda < 1$. The average centroid norms and sphere radii are then respectively controlled by the scalars $r_0$ and $r$, so that for all $ \mu $ and $ x\in M^\mu $, we have that $ x= r_0 u^\mu_0 + r \sum_{i=1}^K s_i u_i^\mu $, with $ \sum_i (s_i)^2 \leq 1 $. The variables $\lambda, \psi$ respectively determine the degree of correlation between the axes and centroids: As $\lambda, \psi \to 1$, the axes and centroids will be fully correlated with one another, while $\lambda,\psi \to 0$ implies randomly oriented axes and centroids [Fig.\;\ref{fig:sphere-comparison}(a)]. 

Even under these simplifying assumptions, the minimization in Eq.\;\eqref{eq:general-T} is not directly solvable. As such, we reframe the problem in terms of a statistical mechanical system with quenched disorder and study the limit $P\to\infty$. To do this, note that the constraint on the fields can be rewritten as $ r_0 V^\mu_0 -r\sqrt{\sum_{i>0}(V_i^\mu)^2 } \geq \kappa $, as can be seen by applying the Karush-Kuhn-Tucker (KKT) conditions \cite{boyd2004convex} to the Lagrangian $ \mathcal L(S, \eta) = r\sum_{i>0}V_i^\mu S_i + \eta (||S||^2 - 1) $ (see SM \cite{note:SM}). The capacity can then be derived by studying the following Gibbs measure: 

\begin{align}
\frac 1 Z \exp\bigg[&-\frac \beta 2\sum_{i,j,\mu,
\nu} y^\mu y^\nu \big(C^{-1}\big)^{\mu, i}_{\nu, j}(V^\mu_i - T^\mu_i)(V^\nu_j - T^\nu_j)\bigg]
\nonumber
\\
&\times \prod_\mu \Theta\bigg(r_0 V^\mu_0 - r\sqrt{\sum_{i>0} (V_i^\mu)^2} - \kappa\bigg)
dV^\mu \; , 
\label{eq:gibbs}
\end{align}

\noindent where $ Z $ is the partition function \cite{gardner1988optimal}. We can see that $1/ \alpha_{cor}(\kappa) $ is then given by the average energy in the zero-temperature limit: $
[{\alpha_{cor}(\kappa)}]^{-1} = -\frac 2 P \lim_{\beta\to\infty} \frac \partial {\partial\beta} \overline{\log Z }$, with the overline denoting the average with respect to the $ T $ and the labels $ y^\mu $. We calculate the resulting free energy density using the replica method---see SM for details \cite{note:SM}. 

Under these assumptions, the capacity is given by 

\begin{gather} 
\frac {1}{\alpha_{cor}(\kappa)} =  K \big(\sqrt q - 1\big)^2  + \int_{-\infty}^{\hat\kappa(q)} \frac{d\xi}{\sqrt{2\pi}} 
e^{
-\frac 1 2\xi^2} 
\big(\xi  - \hat\kappa(q)\big)^2  , 
\label{eq:spheres}
\end{gather}

\noindent where $ q $ is the scaled squared norm of the signed fields of an arbitrary sphere, $ q\equiv \overline{{\sum_{i>0} (V^\mu_i)^2}}/{(K(1 - \lambda))} $, and $\hat \kappa(q)$ is an effective margin. The values of the $q$ and $ \hat \kappa(q)$ are then fixed by the self-consistent equations 

\begin{gather}
\sqrt{q}  =  1 +  \frac{r\sqrt{1 - \lambda}}{r_0\sqrt{K(1 - \psi)}}
\int_{-\infty}^{\hat\kappa(q)}\frac{d\xi}{\sqrt{2\pi}} 
e^{
-\frac 1 2\xi^2} 
\big(\xi  - \hat\kappa(q) \big) \; ,
\nonumber
\\
\hat\kappa(q) = \frac{r\sqrt{K(1-\lambda)q} + \kappa}{r_0\sqrt{1 - \psi}} \; .
\label{eq:self-const}
\end{gather}

\begin{figure}
\includegraphics[width=8.6cm]{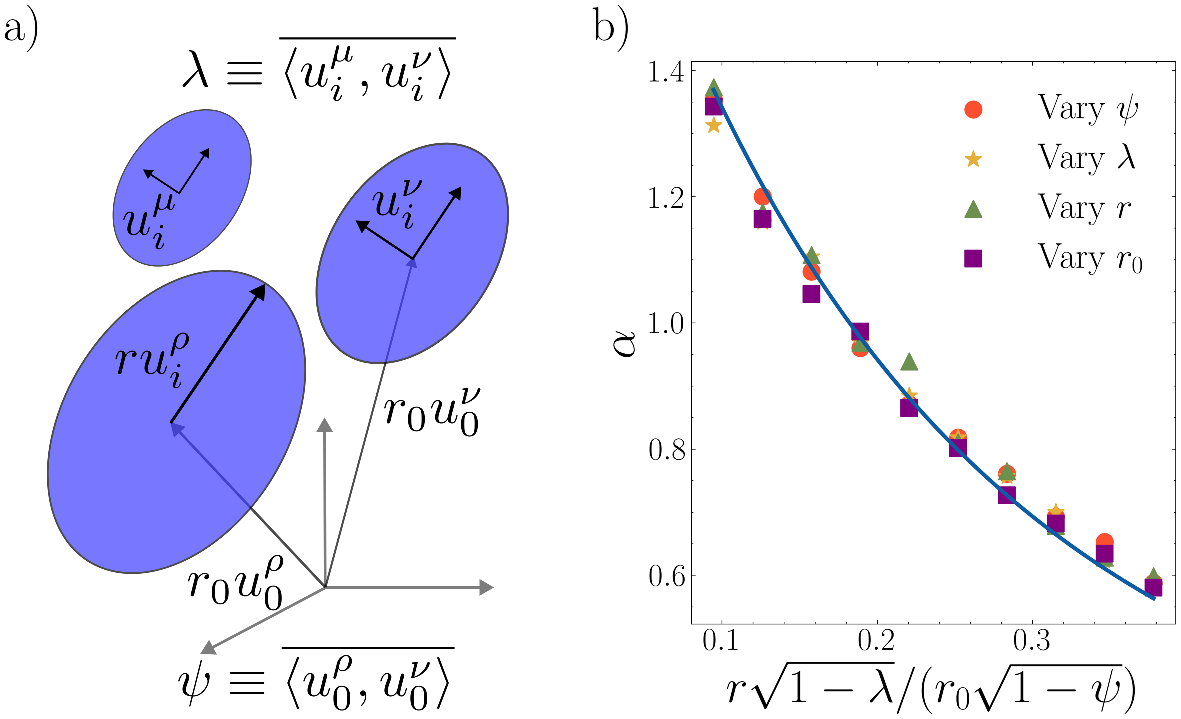}
\caption{\label{fig:sphere-comparison} 
The capacity for correlated spheres. (a) Visual demonstration of spherical manifolds with low-rank axis and centroid correlations. (b) The zero-margin capacity as a function of only the input ratio $r\sqrt{1-\lambda}/(r_0\sqrt{1-\psi})$. Points represent averages over five random sphere samplings, and the solid line represents the theoretical prediction. For each experiment, we fix three of the four parameters and vary the remaining one to obtain a fixed value of the ratio.
}

\end{figure} 

\noindent With our definition of $ q $ and $\hat\kappa (q)$ in hand, we can see that the capacity for correlated spheres is the same as the capacity of random points given in Eq.\;\eqref{eq:cap-pts} with an effective margin of $ \hat\kappa(q) $, plus an extra bias term which corresponds to additional contributions to the capacity from the correlations and spread of the spheres.

The above solution gives a direct view into the effects of correlations on manifold separability. From Eqs.\;\eqref{eq:spheres} and \eqref{eq:self-const}, we can see that when $\kappa=0$, both $q$ and the effective margin are fully determined by the ratio $r\sqrt{(1-\lambda)}/(r_0\sqrt{1-\psi})$ (Fig.\;\ref{fig:sphere-comparison}). Even when $\kappa\neq 0$, the sphere radii and centroid scalings, $ r, r_0 $, and the respective correlations, $ \lambda, \psi, $ only affect the capacity through the products: $r\sqrt{1-\lambda}$, $r_0\sqrt{1-\psi}$. This implies that increasing the axis or centroid correlations affects the capacity in the same way as shrinking the spheres or centroid norms does. That is, axis correlations effectively shrink the sphere radii, while centroid correlations effectively push the manifolds closer to the origin. 

These effects are most dramatic when we consider the limits of fully correlated manifolds. In the fully correlated centroids limit, $ \psi \to 1$, we can see that the capacity falls to 0. Conversely, in the fully correlated axes limit, $ \lambda\to1 $, we can see that $ \sqrt{q} \to 1 $, so that the capacity grows to the capacity for random points with margin $ \kappa / (r_0 \sqrt{1 -\psi})$ \cite{gardner_space_1988}. This shows that high-dimensional, fully correlated spheres are as easy to separate as randomly oriented points---see \cite{farrell2022capacity, Lopez_1995} for related results. 

\begin{figure}
\includegraphics[width=8.6cm]{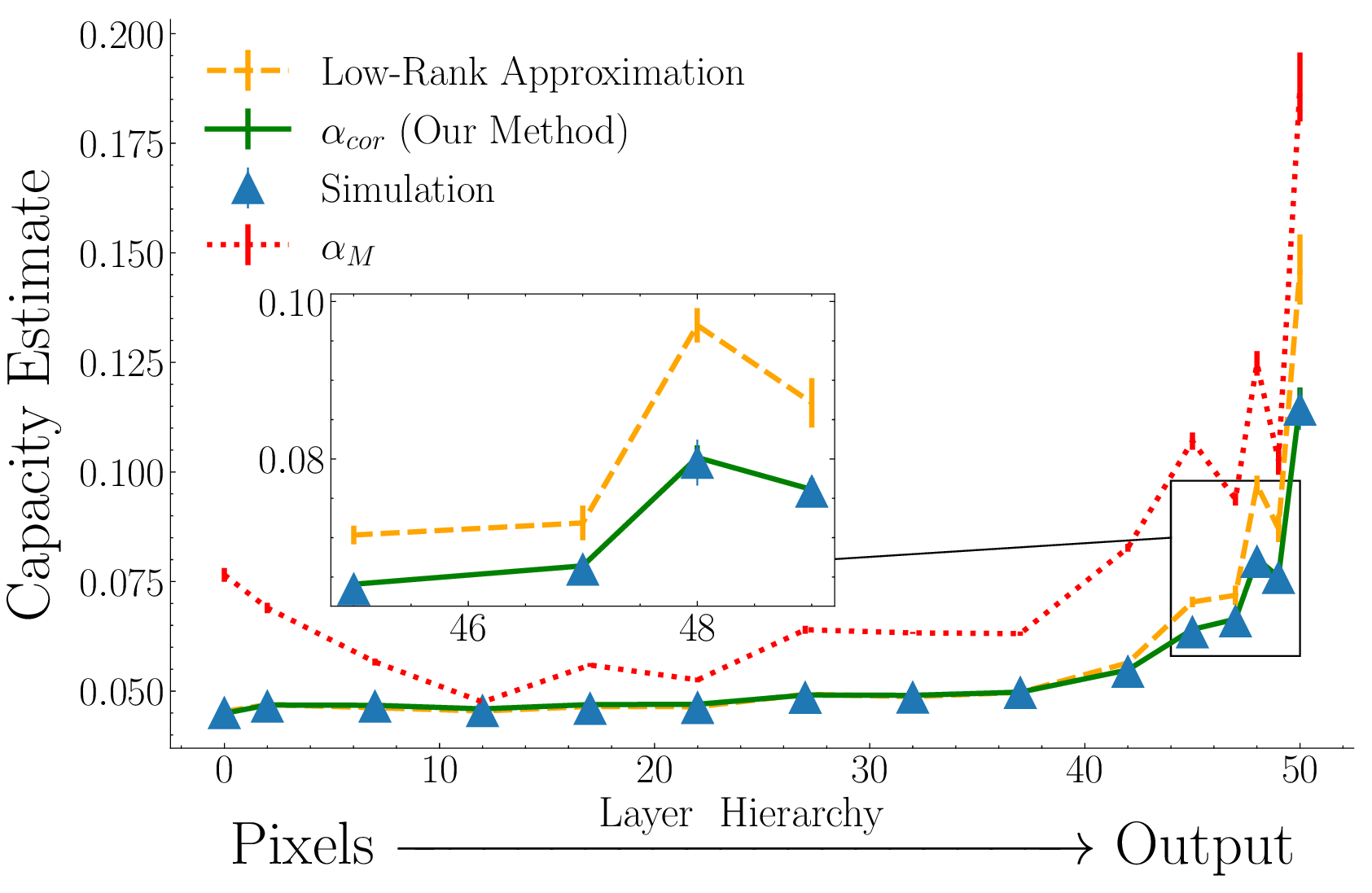}
\caption{\label{fig:simclr} 
Comparison of the low-rank approximation (yellow dashed line) \cite{cohen_separability_2020}, $\alpha_M$ (red dotted line) \cite{chung_classification_2018}, and our $\alpha_{cor}$ calculation (green solid line) to the ground truth simulation capacity (blue triangles) \cite{chung_linear_2016} on data manifolds arising from the ResNet50 artificial neural network architecture trained using SimCLR on the ImageNet dataset \cite{chen2020simple, he2016deep}.}
\end{figure} 

\emph{Application to deep network manifolds:} Having studied our theoretical predictions in two simple settings, we now consider the performance of our capacity estimator, $\alpha_{cor}$, when applied to neural manifolds from a pretrained SimCLR ResNet50 network on the ImageNet dataset \cite{chen2020simple, ILSVRC15, he2016deep}. We can see from Fig.\;\ref{fig:simclr} that the low rank approximation \cite{cohen_separability_2020} significantly overestimates the capacity in later layers of the network. Note that while we can numerically estimate the ground truth simulation capacity here because we use few data points (see SM; \cite{note:SM}), this is computationally infeasible for larger data manifolds \cite{cohen_separability_2020, chung_learning_2018}. Thus, our $\alpha_{cor}$ estimator can be used to estimate the capacity where other methods fail. 

\emph{Discussion:} 
In this Letter, we considered the problem of linearly separating a set of high-dimensional manifolds whose centroids and axes are correlated with one another. We first derived an expression for the capacity of general manifolds with arbitrary covariance tensors. After showing that the resulting expression outperforms previous capacity estimators when presented with correlated manifolds, we turned to the problem of interpreting the resulting expression for the capacity. To this end, we considered the problem of linearly separating spheres with homogeneous correlations along the centroids and axes. The resulting expression for the capacity closely tracks the capacity for points with an effective margin determined by the geometry and correlations of the spheres. Remarkably, we found that centroid and axis correlations play the same roles as the distance of the spheres from the origin and the sphere radii, respectively. These findings reveal a duality between representational geometry and correlations with respect to the problem of classification. 

Our work suggests two main subsequent lines of research. First, given the rising popularity and sophistication of geometric analysis methods in neuroscience \cite{chaudhuri_intrinsic_2019, bernardi_geometry_2020, chung_neural_2021, cohen_separability_2020, chung_classification_2018}, together with the extensive literature examining the phenomenology and role of different types of neural correlations \cite{panzeri_structures_2022, zylberberg_robust_2017}, we hope to apply the results from this study to further connect these two lines of inquiry. One particularly interesting approach in this direction would be to apply our results to study the relationship between hierarchical correlation structures, geometry, and the organization of abstract knowledge, especially in the context of multilabel classification \cite{saxe_mathematical_2019, johnston2023abstract, bernardi_geometry_2020}. Another interesting approach would be to use Eq.\;\eqref{eq:general-T} to  derive a set of metrics quantifying the effects of different types of neural correlations on the capacity for arbitrary data manifolds, complementing preexisting measures describing the impact of geometry on the capacity \cite{chung_classification_2018, cohen_separability_2020}. 

Second, our results regarding spheres with correlated axes suggest that self-supervised objectives which produce positive correlations between manifold axes could yield latent representations with favorable classification properties. If we further define manifold axes using the translation between an original image and its augmentation, such an objective could also produce representations which are disentangled with respect to, for example, color distortion and rotation \cite{higgins2017betavae, higgins_towards_2018}. We hope to pursue this line of research in subsequent work. 

\begin{acknowledgments}
    {\bf{Acknowledgments: }}The authors thank Abdulkadir Canatar and Chi-Ning Chou for their comments on an earlier version of this manuscript.
\end{acknowledgments}


\begin{thebibliography}{47}%
\makeatletter
\providecommand \@ifxundefined [1]{%
 \@ifx{#1\undefined}
}%
\providecommand \@ifnum [1]{%
 \ifnum #1\expandafter \@firstoftwo
 \else \expandafter \@secondoftwo
 \fi
}%
\providecommand \@ifx [1]{%
 \ifx #1\expandafter \@firstoftwo
 \else \expandafter \@secondoftwo
 \fi
}%
\providecommand \natexlab [1]{#1}%
\providecommand \enquote  [1]{``#1''}%
\providecommand \bibnamefont  [1]{#1}%
\providecommand \bibfnamefont [1]{#1}%
\providecommand \citenamefont [1]{#1}%
\providecommand \href@noop [0]{\@secondoftwo}%
\providecommand \href [0]{\begingroup \@sanitize@url \@href}%
\providecommand \@href[1]{\@@startlink{#1}\@@href}%
\providecommand \@@href[1]{\endgroup#1\@@endlink}%
\providecommand \@sanitize@url [0]{\catcode `\\12\catcode `\$12\catcode
  `\&12\catcode `\#12\catcode `\^12\catcode `\_12\catcode `\%12\relax}%
\providecommand \@@startlink[1]{}%
\providecommand \@@endlink[0]{}%
\providecommand \url  [0]{\begingroup\@sanitize@url \@url }%
\providecommand \@url [1]{\endgroup\@href {#1}{\urlprefix }}%
\providecommand \urlprefix  [0]{URL }%
\providecommand \Eprint [0]{\href }%
\providecommand \doibase [0]{https://doi.org/}%
\providecommand \selectlanguage [0]{\@gobble}%
\providecommand \bibinfo  [0]{\@secondoftwo}%
\providecommand \bibfield  [0]{\@secondoftwo}%
\providecommand \translation [1]{[#1]}%
\providecommand \BibitemOpen [0]{}%
\providecommand \bibitemStop [0]{}%
\providecommand \bibitemNoStop [0]{.\EOS\space}%
\providecommand \EOS [0]{\spacefactor3000\relax}%
\providecommand \BibitemShut  [1]{\csname bibitem#1\endcsname}%
\let\auto@bib@innerbib\@empty
\bibitem [{\citenamefont {Rotondo}\ \emph {et~al.}(2020)\citenamefont
  {Rotondo}, \citenamefont {Lagomarsino},\ and\ \citenamefont
  {Gherardi}}]{rotondo_counting_2020}%
  \BibitemOpen
  \bibfield  {author} {\bibinfo {author} {\bibfnamefont {P.}~\bibnamefont
  {Rotondo}}, \bibinfo {author} {\bibfnamefont {M.~C.}\ \bibnamefont
  {Lagomarsino}},\ and\ \bibinfo {author} {\bibfnamefont {M.}~\bibnamefont
  {Gherardi}},\ }\bibfield  {title} {\bibinfo {title} {Counting the learnable
  functions of geometrically structured data},\ }\href
  {https://doi.org/10.1103/PhysRevResearch.2.023169} {\bibfield  {journal}
  {\bibinfo  {journal} {Physical Review Research}\ }\textbf {\bibinfo {volume}
  {2}},\ \bibinfo {pages} {023169} (\bibinfo {year} {2020})}\BibitemShut
  {NoStop}%
\bibitem [{\citenamefont {Battista}\ and\ \citenamefont
  {Monasson}(2020)}]{battista_capacity-resolution_2020}%
  \BibitemOpen
  \bibfield  {author} {\bibinfo {author} {\bibfnamefont {A.}~\bibnamefont
  {Battista}}\ and\ \bibinfo {author} {\bibfnamefont {R.}~\bibnamefont
  {Monasson}},\ }\bibfield  {title} {\bibinfo {title} {Capacity-{Resolution}
  {Trade}-{Off} in the {Optimal} {Learning} of {Multiple} {Low}-{Dimensional}
  {Manifolds} by {Attractor} {Neural} {Networks}},\ }\href
  {https://doi.org/10.1103/PhysRevLett.124.048302} {\bibfield  {journal}
  {\bibinfo  {journal} {Physical Review Letters}\ }\textbf {\bibinfo {volume}
  {124}},\ \bibinfo {pages} {048302} (\bibinfo {year} {2020})}\BibitemShut
  {NoStop}%
\bibitem [{\citenamefont {Goldt}\ \emph {et~al.}(2020)\citenamefont {Goldt},
  \citenamefont {Mézard}, \citenamefont {Krzakala},\ and\ \citenamefont
  {Zdeborová}}]{goldt_modeling_2020}%
  \BibitemOpen
  \bibfield  {author} {\bibinfo {author} {\bibfnamefont {S.}~\bibnamefont
  {Goldt}}, \bibinfo {author} {\bibfnamefont {M.}~\bibnamefont {Mézard}},
  \bibinfo {author} {\bibfnamefont {F.}~\bibnamefont {Krzakala}},\ and\
  \bibinfo {author} {\bibfnamefont {L.}~\bibnamefont {Zdeborová}},\ }\bibfield
   {title} {\bibinfo {title} {Modeling the {Influence} of {Data} {Structure} on
  {Learning} in {Neural} {Networks}: {The} {Hidden} {Manifold} {Model}},\
  }\href {https://doi.org/10.1103/PhysRevX.10.041044} {\bibfield  {journal}
  {\bibinfo  {journal} {Physical Review X}\ }\textbf {\bibinfo {volume} {10}},\
  \bibinfo {pages} {041044} (\bibinfo {year} {2020})}\BibitemShut {NoStop}%
\bibitem [{\citenamefont {Farrell}\ \emph {et~al.}(2022)\citenamefont
  {Farrell}, \citenamefont {Bordelon}, \citenamefont {Trivedi},\ and\
  \citenamefont {Pehlevan}}]{farrell2022capacity}%
  \BibitemOpen
  \bibfield  {author} {\bibinfo {author} {\bibfnamefont {M.}~\bibnamefont
  {Farrell}}, \bibinfo {author} {\bibfnamefont {B.}~\bibnamefont {Bordelon}},
  \bibinfo {author} {\bibfnamefont {S.}~\bibnamefont {Trivedi}},\ and\ \bibinfo
  {author} {\bibfnamefont {C.}~\bibnamefont {Pehlevan}},\ }\bibfield  {title}
  {\bibinfo {title} {Capacity of group-invariant linear readouts from
  equivariant representations: How many objects can be linearly classified
  under all possible views?},\ }in\ \href
  {https://openreview.net/forum?id=_4GFbtOuWq-} {\emph {\bibinfo {booktitle}
  {International Conference on Learning Representations}}}\ (\bibinfo {year}
  {2022})\BibitemShut {NoStop}%
\bibitem [{\citenamefont {Biswas}\ and\ \citenamefont
  {Fitzgerald}(2022)}]{biswas_geometric_2022}%
  \BibitemOpen
  \bibfield  {author} {\bibinfo {author} {\bibfnamefont {T.}~\bibnamefont
  {Biswas}}\ and\ \bibinfo {author} {\bibfnamefont {J.~E.}\ \bibnamefont
  {Fitzgerald}},\ }\bibfield  {title} {\bibinfo {title} {Geometric framework to
  predict structure from function in neural networks},\ }\href
  {https://doi.org/10.1103/PhysRevResearch.4.023255} {\bibfield  {journal}
  {\bibinfo  {journal} {Physical Review Research}\ }\textbf {\bibinfo {volume}
  {4}},\ \bibinfo {pages} {023255} (\bibinfo {year} {2022})}\BibitemShut
  {NoStop}%
\bibitem [{\citenamefont {Susman}\ \emph {et~al.}(2021)\citenamefont {Susman},
  \citenamefont {Mastrogiuseppe}, \citenamefont {Brenner},\ and\ \citenamefont
  {Barak}}]{susman_quality_2021}%
  \BibitemOpen
  \bibfield  {author} {\bibinfo {author} {\bibfnamefont {L.}~\bibnamefont
  {Susman}}, \bibinfo {author} {\bibfnamefont {F.}~\bibnamefont
  {Mastrogiuseppe}}, \bibinfo {author} {\bibfnamefont {N.}~\bibnamefont
  {Brenner}},\ and\ \bibinfo {author} {\bibfnamefont {O.}~\bibnamefont
  {Barak}},\ }\bibfield  {title} {\bibinfo {title} {Quality of internal
  representation shapes learning performance in feedback neural networks},\
  }\href {https://doi.org/10.1103/PhysRevResearch.3.013176} {\bibfield
  {journal} {\bibinfo  {journal} {Physical Review Research}\ }\textbf {\bibinfo
  {volume} {3}},\ \bibinfo {pages} {013176} (\bibinfo {year}
  {2021})}\BibitemShut {NoStop}%
\bibitem [{\citenamefont {Ansuini}\ \emph {et~al.}(2019)\citenamefont
  {Ansuini}, \citenamefont {Laio}, \citenamefont {Macke},\ and\ \citenamefont
  {Zoccolan}}]{NEURIPS2019_cfcce062}%
  \BibitemOpen
  \bibfield  {author} {\bibinfo {author} {\bibfnamefont {A.}~\bibnamefont
  {Ansuini}}, \bibinfo {author} {\bibfnamefont {A.}~\bibnamefont {Laio}},
  \bibinfo {author} {\bibfnamefont {J.~H.}\ \bibnamefont {Macke}},\ and\
  \bibinfo {author} {\bibfnamefont {D.}~\bibnamefont {Zoccolan}},\ }\bibfield
  {title} {\bibinfo {title} {Intrinsic dimension of data representations in
  deep neural networks},\ }in\ \href
  {https://proceedings.neurips.cc/paper/2019/file/cfcce0621b49c983991ead4c3d4d3b6b-Paper.pdf}
  {\emph {\bibinfo {booktitle} {Advances in neural information processing
  systems}}},\ Vol.~\bibinfo {volume} {32},\ \bibinfo {editor} {edited by\
  \bibinfo {editor} {\bibfnamefont {H.}~\bibnamefont {Wallach}}, \bibinfo
  {editor} {\bibfnamefont {H.}~\bibnamefont {Larochelle}}, \bibinfo {editor}
  {\bibfnamefont {A.}~\bibnamefont {Beygelzimer}}, \bibinfo {editor}
  {\bibfnamefont {F.}~\bibnamefont {dAlché Buc}}, \bibinfo {editor}
  {\bibfnamefont {E.}~\bibnamefont {Fox}},\ and\ \bibinfo {editor}
  {\bibfnamefont {R.}~\bibnamefont {Garnett}}}\ (\bibinfo  {publisher} {Curran
  Associates, Inc.},\ \bibinfo {year} {2019})\BibitemShut {NoStop}%
\bibitem [{\citenamefont {Dahmen}\ \emph {et~al.}(2020)\citenamefont {Dahmen},
  \citenamefont {Gilson},\ and\ \citenamefont {Helias}}]{dahmen_capacity_2020}%
  \BibitemOpen
  \bibfield  {author} {\bibinfo {author} {\bibfnamefont {D.}~\bibnamefont
  {Dahmen}}, \bibinfo {author} {\bibfnamefont {M.}~\bibnamefont {Gilson}},\
  and\ \bibinfo {author} {\bibfnamefont {M.}~\bibnamefont {Helias}},\
  }\bibfield  {title} {\bibinfo {title} {Capacity of the covariance
  perceptron},\ }\href {https://doi.org/10.1088/1751-8121/ab82dd} {\bibfield
  {journal} {\bibinfo  {journal} {Journal of Physics A: Mathematical and
  Theoretical}\ }\textbf {\bibinfo {volume} {53}},\ \bibinfo {pages} {354002}
  (\bibinfo {year} {2020})}\BibitemShut {NoStop}%
\bibitem [{\citenamefont {Steinberg}\ and\ \citenamefont
  {Sompolinsky}(2022)}]{steinberg2022associative}%
  \BibitemOpen
  \bibfield  {author} {\bibinfo {author} {\bibfnamefont {J.}~\bibnamefont
  {Steinberg}}\ and\ \bibinfo {author} {\bibfnamefont {H.}~\bibnamefont
  {Sompolinsky}},\ }\bibfield  {title} {\bibinfo {title} {Associative memory of
  structured knowledge},\ }\href@noop {} {\bibfield  {journal} {\bibinfo
  {journal} {Scientific Reports}\ }\textbf {\bibinfo {volume} {12}},\ \bibinfo
  {pages} {21808} (\bibinfo {year} {2022})}\BibitemShut {NoStop}%
\bibitem [{\citenamefont {Cohen}\ and\ \citenamefont
  {Sompolinsky}(2022)}]{cohen_soft-margin_2022}%
  \BibitemOpen
  \bibfield  {author} {\bibinfo {author} {\bibfnamefont {U.}~\bibnamefont
  {Cohen}}\ and\ \bibinfo {author} {\bibfnamefont {H.}~\bibnamefont
  {Sompolinsky}},\ }\bibfield  {title} {\bibinfo {title} {Soft-margin
  classification of object manifolds},\ }\href
  {https://doi.org/10.1103/PhysRevE.106.024126} {\bibfield  {journal} {\bibinfo
   {journal} {Physical Review E}\ }\textbf {\bibinfo {volume} {106}},\ \bibinfo
  {pages} {024126} (\bibinfo {year} {2022})},\ \bibinfo {note}
  {arXiv:2203.07040 [cond-mat, q-bio, stat]}\BibitemShut {NoStop}%
\bibitem [{\citenamefont {Chaudhuri}\ \emph {et~al.}(2019)\citenamefont
  {Chaudhuri}, \citenamefont {Gerçek}, \citenamefont {Pandey}, \citenamefont
  {Peyrache},\ and\ \citenamefont {Fiete}}]{chaudhuri_intrinsic_2019}%
  \BibitemOpen
  \bibfield  {author} {\bibinfo {author} {\bibfnamefont {R.}~\bibnamefont
  {Chaudhuri}}, \bibinfo {author} {\bibfnamefont {B.}~\bibnamefont {Gerçek}},
  \bibinfo {author} {\bibfnamefont {B.}~\bibnamefont {Pandey}}, \bibinfo
  {author} {\bibfnamefont {A.}~\bibnamefont {Peyrache}},\ and\ \bibinfo
  {author} {\bibfnamefont {I.}~\bibnamefont {Fiete}},\ }\bibfield  {title}
  {\bibinfo {title} {The intrinsic attractor manifold and population dynamics
  of a canonical cognitive circuit across waking and sleep},\ }\href
  {https://doi.org/10.1038/s41593-019-0460-x} {\bibfield  {journal} {\bibinfo
  {journal} {Nature Neuroscience}\ }\textbf {\bibinfo {volume} {22}},\ \bibinfo
  {pages} {1512} (\bibinfo {year} {2019})}\BibitemShut {NoStop}%
\bibitem [{\citenamefont {Bernardi}\ \emph {et~al.}(2020)\citenamefont
  {Bernardi}, \citenamefont {Benna}, \citenamefont {Rigotti}, \citenamefont
  {Munuera}, \citenamefont {Fusi},\ and\ \citenamefont
  {Salzman}}]{bernardi_geometry_2020}%
  \BibitemOpen
  \bibfield  {author} {\bibinfo {author} {\bibfnamefont {S.}~\bibnamefont
  {Bernardi}}, \bibinfo {author} {\bibfnamefont {M.~K.}\ \bibnamefont {Benna}},
  \bibinfo {author} {\bibfnamefont {M.}~\bibnamefont {Rigotti}}, \bibinfo
  {author} {\bibfnamefont {J.}~\bibnamefont {Munuera}}, \bibinfo {author}
  {\bibfnamefont {S.}~\bibnamefont {Fusi}},\ and\ \bibinfo {author}
  {\bibfnamefont {C.~D.}\ \bibnamefont {Salzman}},\ }\bibfield  {title}
  {\bibinfo {title} {The {Geometry} of {Abstraction} in the {Hippocampus} and
  {Prefrontal} {Cortex}},\ }\href {https://doi.org/10.1016/j.cell.2020.09.031}
  {\bibfield  {journal} {\bibinfo  {journal} {Cell}\ }\textbf {\bibinfo
  {volume} {183}},\ \bibinfo {pages} {954} (\bibinfo {year} {2020})},\ \bibinfo
  {note} {publisher: Elsevier}\BibitemShut {NoStop}%
\bibitem [{\citenamefont {Chung}\ and\ \citenamefont
  {Abbott}(2021)}]{chung_neural_2021}%
  \BibitemOpen
  \bibfield  {author} {\bibinfo {author} {\bibfnamefont {S.}~\bibnamefont
  {Chung}}\ and\ \bibinfo {author} {\bibfnamefont {L.}~\bibnamefont {Abbott}},\
  }\bibfield  {title} {\bibinfo {title} {Neural population geometry: An
  approach for understanding biological and artificial neural networks},\
  }\href@noop {} {\bibfield  {journal} {\bibinfo  {journal} {Current opinion in
  neurobiology}\ }\textbf {\bibinfo {volume} {70}},\ \bibinfo {pages} {137}
  (\bibinfo {year} {2021})}\BibitemShut {NoStop}%
\bibitem [{\citenamefont {Sorscher}\ \emph {et~al.}(2022)\citenamefont
  {Sorscher}, \citenamefont {Ganguli},\ and\ \citenamefont
  {Sompolinsky}}]{sorscher_neural_2022}%
  \BibitemOpen
  \bibfield  {author} {\bibinfo {author} {\bibfnamefont {B.}~\bibnamefont
  {Sorscher}}, \bibinfo {author} {\bibfnamefont {S.}~\bibnamefont {Ganguli}},\
  and\ \bibinfo {author} {\bibfnamefont {H.}~\bibnamefont {Sompolinsky}},\
  }\bibfield  {title} {\bibinfo {title} {Neural representational geometry
  underlies few-shot concept learning},\ }\href
  {https://doi.org/10.1073/pnas.2200800119} {\bibfield  {journal} {\bibinfo
  {journal} {Proceedings of the National Academy of Sciences}\ }\textbf
  {\bibinfo {volume} {119}},\ \bibinfo {pages} {e2200800119} (\bibinfo {year}
  {2022})}\BibitemShut {NoStop}%
\bibitem [{\citenamefont {Chung}\ \emph
  {et~al.}(2018{\natexlab{a}})\citenamefont {Chung}, \citenamefont {Lee},\ and\
  \citenamefont {Sompolinsky}}]{chung_classification_2018}%
  \BibitemOpen
  \bibfield  {author} {\bibinfo {author} {\bibfnamefont {S.}~\bibnamefont
  {Chung}}, \bibinfo {author} {\bibfnamefont {D.~D.}\ \bibnamefont {Lee}},\
  and\ \bibinfo {author} {\bibfnamefont {H.}~\bibnamefont {Sompolinsky}},\
  }\bibfield  {title} {\bibinfo {title} {Classification and {Geometry} of
  {General} {Perceptual} {Manifolds}},\ }\href
  {https://doi.org/10.1103/PhysRevX.8.031003} {\bibfield  {journal} {\bibinfo
  {journal} {Physical Review X}\ }\textbf {\bibinfo {volume} {8}},\ \bibinfo
  {pages} {031003} (\bibinfo {year} {2018}{\natexlab{a}})}\BibitemShut
  {NoStop}%
\bibitem [{\citenamefont {Cohen}\ \emph {et~al.}(2020)\citenamefont {Cohen},
  \citenamefont {Chung}, \citenamefont {Lee},\ and\ \citenamefont
  {Sompolinsky}}]{cohen_separability_2020}%
  \BibitemOpen
  \bibfield  {author} {\bibinfo {author} {\bibfnamefont {U.}~\bibnamefont
  {Cohen}}, \bibinfo {author} {\bibfnamefont {S.}~\bibnamefont {Chung}},
  \bibinfo {author} {\bibfnamefont {D.~D.}\ \bibnamefont {Lee}},\ and\ \bibinfo
  {author} {\bibfnamefont {H.}~\bibnamefont {Sompolinsky}},\ }\bibfield
  {title} {\bibinfo {title} {Separability and geometry of object manifolds in
  deep neural networks},\ }\href {https://doi.org/10.1038/s41467-020-14578-5}
  {\bibfield  {journal} {\bibinfo  {journal} {Nature Communications}\ }\textbf
  {\bibinfo {volume} {11}},\ \bibinfo {pages} {746} (\bibinfo {year} {2020})},\
  \bibinfo {note} {number: 1 Publisher: Nature Publishing Group}\BibitemShut
  {NoStop}%
\bibitem [{\citenamefont {Chung}\ \emph {et~al.}(2016)\citenamefont {Chung},
  \citenamefont {Lee},\ and\ \citenamefont {Sompolinsky}}]{chung_linear_2016}%
  \BibitemOpen
  \bibfield  {author} {\bibinfo {author} {\bibfnamefont {S.}~\bibnamefont
  {Chung}}, \bibinfo {author} {\bibfnamefont {D.~D.}\ \bibnamefont {Lee}},\
  and\ \bibinfo {author} {\bibfnamefont {H.}~\bibnamefont {Sompolinsky}},\
  }\bibfield  {title} {\bibinfo {title} {Linear {Readout} of {Object}
  {Manifolds}},\ }\href {https://doi.org/10.1103/PhysRevE.93.060301} {\bibfield
   {journal} {\bibinfo  {journal} {Physical Review E}\ }\textbf {\bibinfo
  {volume} {93}},\ \bibinfo {pages} {060301} (\bibinfo {year} {2016})},\
  \bibinfo {note} {arXiv:1512.01834 [cond-mat, q-bio, stat]}\BibitemShut
  {NoStop}%
\bibitem [{\citenamefont {Harrison}(2011)}]{birdimg}%
  \BibitemOpen
  \bibfield  {author} {\bibinfo {author} {\bibfnamefont {J.~J.}\ \bibnamefont
  {Harrison}},\ }\href@noop {} {\bibinfo {title} {{Azure} {Kingfisher}}}
  (\bibinfo {year} {2011}),\ \bibinfo {note}
  {\url{https://upload.wikimedia.org/wikipedia/commons/7/72/Alcedo_azurea_-_Julatten.jpg}
  This work is licensed under the Creative Commons 3.0 Unported License. To
  view a copy of this license, visit
  \url{https://creativecommons.org/licenses/by/3.0/legalcode}}\BibitemShut
  {NoStop}%
\bibitem [{\citenamefont {Korneev}(2021)}]{vaseimg}%
  \BibitemOpen
  \bibfield  {author} {\bibinfo {author} {\bibfnamefont {S.}~\bibnamefont
  {Korneev}},\ }\href@noop {} {\bibinfo {title} {{Faberge} {Vase}}} (\bibinfo
  {year} {2021}),\ \bibinfo {note}
  {\url{https://upload.wikimedia.org/wikipedia/commons/9/9e/Faberge_vase_State_Museum_of_Sport_1928.jpg}
  This work is licensed under the Creative Commons 4.0 ShareAlike License
  International. To view a copy of this license, visit
  \url{https://creativecommons.org/licenses/by-sa/4.0/legalcode}}\BibitemShut
  {NoStop}%
\bibitem [{\citenamefont {Karwath}(2005)}]{chinesevaseimg}%
  \BibitemOpen
  \bibfield  {author} {\bibinfo {author} {\bibfnamefont {A.}~\bibnamefont
  {Karwath}},\ }\href@noop {} {\bibinfo {title} {{Hand-made} {Chinese} {Vase}}}
  (\bibinfo {year} {2005}),\ \bibinfo {note}
  {\url{https://upload.wikimedia.org/wikipedia/commons/b/b8/Chinese_vase.jpg}
  This work is licensed under the Creative Commons 2.5 Generic ShareAlike
  License. To view a copy of this license, visit
  \url{https://creativecommons.org/licenses/by-sa/2.5/legalcode}}\BibitemShut
  {NoStop}%
\bibitem [{\citenamefont {Fioreze}(2008)}]{skyimg}%
  \BibitemOpen
  \bibfield  {author} {\bibinfo {author} {\bibfnamefont {T.}~\bibnamefont
  {Fioreze}},\ }\href@noop {} {\bibinfo {title} {{Clouds} over the {Atlantic}
  {Ocean}}} (\bibinfo {year} {2008}),\ \bibinfo {note}
  {\url{https://upload.wikimedia.org/wikipedia/commons/e/e0/Clouds_over_the_Atlantic_Ocean.jpg}
  This work is licensed under the Creative Commons ShareAlike 3.0 Unported
  License. To view a copy of this license, visit
  \url{https://creativecommons.org/licenses/by-sa/3.0/legalcode}}\BibitemShut
  {NoStop}%
\bibitem [{\citenamefont {Panzeri}\ \emph {et~al.}(2022)\citenamefont
  {Panzeri}, \citenamefont {Moroni}, \citenamefont {Safaai},\ and\
  \citenamefont {Harvey}}]{panzeri_structures_2022}%
  \BibitemOpen
  \bibfield  {author} {\bibinfo {author} {\bibfnamefont {S.}~\bibnamefont
  {Panzeri}}, \bibinfo {author} {\bibfnamefont {M.}~\bibnamefont {Moroni}},
  \bibinfo {author} {\bibfnamefont {H.}~\bibnamefont {Safaai}},\ and\ \bibinfo
  {author} {\bibfnamefont {C.~D.}\ \bibnamefont {Harvey}},\ }\bibfield  {title}
  {\bibinfo {title} {The structures and functions of correlations in neural
  population codes},\ }\href {https://doi.org/10.1038/s41583-022-00606-4}
  {\bibfield  {journal} {\bibinfo  {journal} {Nature Reviews Neuroscience}\
  }\textbf {\bibinfo {volume} {23}},\ \bibinfo {pages} {551} (\bibinfo {year}
  {2022})}\BibitemShut {NoStop}%
\bibitem [{\citenamefont {Zylberberg}\ \emph {et~al.}(2017)\citenamefont
  {Zylberberg}, \citenamefont {Pouget}, \citenamefont {Latham},\ and\
  \citenamefont {Shea-Brown}}]{zylberberg_robust_2017}%
  \BibitemOpen
  \bibfield  {author} {\bibinfo {author} {\bibfnamefont {J.}~\bibnamefont
  {Zylberberg}}, \bibinfo {author} {\bibfnamefont {A.}~\bibnamefont {Pouget}},
  \bibinfo {author} {\bibfnamefont {P.~E.}\ \bibnamefont {Latham}},\ and\
  \bibinfo {author} {\bibfnamefont {E.}~\bibnamefont {Shea-Brown}},\ }\bibfield
   {title} {\bibinfo {title} {Robust information propagation through noisy
  neural circuits},\ }\href {https://doi.org/10.1371/journal.pcbi.1005497}
  {\bibfield  {journal} {\bibinfo  {journal} {PLOS Computational Biology}\
  }\textbf {\bibinfo {volume} {13}},\ \bibinfo {pages} {e1005497} (\bibinfo
  {year} {2017})}\BibitemShut {NoStop}%
\bibitem [{\citenamefont {Morcos}\ \emph {et~al.}(2018)\citenamefont {Morcos},
  \citenamefont {Raghu},\ and\ \citenamefont {Bengio}}]{morcos2018insights}%
  \BibitemOpen
  \bibfield  {author} {\bibinfo {author} {\bibfnamefont {A.}~\bibnamefont
  {Morcos}}, \bibinfo {author} {\bibfnamefont {M.}~\bibnamefont {Raghu}},\ and\
  \bibinfo {author} {\bibfnamefont {S.}~\bibnamefont {Bengio}},\ }\bibfield
  {title} {\bibinfo {title} {Insights on representational similarity in neural
  networks with canonical correlation},\ }\href@noop {} {\bibfield  {journal}
  {\bibinfo  {journal} {Advances in Neural Information Processing Systems}\
  }\textbf {\bibinfo {volume} {31}} (\bibinfo {year} {2018})}\BibitemShut
  {NoStop}%
\bibitem [{\citenamefont {Kornblith}\ \emph {et~al.}(2019)\citenamefont
  {Kornblith}, \citenamefont {Norouzi}, \citenamefont {Lee},\ and\
  \citenamefont {Hinton}}]{kornblith2019similarity}%
  \BibitemOpen
  \bibfield  {author} {\bibinfo {author} {\bibfnamefont {S.}~\bibnamefont
  {Kornblith}}, \bibinfo {author} {\bibfnamefont {M.}~\bibnamefont {Norouzi}},
  \bibinfo {author} {\bibfnamefont {H.}~\bibnamefont {Lee}},\ and\ \bibinfo
  {author} {\bibfnamefont {G.}~\bibnamefont {Hinton}},\ }\bibfield  {title}
  {\bibinfo {title} {Similarity of neural network representations revisited},\
  }in\ \href@noop {} {\emph {\bibinfo {booktitle} {International conference on
  machine learning}}}\ (\bibinfo {year} {2019})\ pp.\ \bibinfo {pages}
  {3519--3529},\ \bibinfo {note} {tex.organization: PMLR}\BibitemShut {NoStop}%
\bibitem [{\citenamefont {Chen}\ \emph {et~al.}(2020)\citenamefont {Chen},
  \citenamefont {Kornblith}, \citenamefont {Norouzi},\ and\ \citenamefont
  {Hinton}}]{chen2020simple}%
  \BibitemOpen
  \bibfield  {author} {\bibinfo {author} {\bibfnamefont {T.}~\bibnamefont
  {Chen}}, \bibinfo {author} {\bibfnamefont {S.}~\bibnamefont {Kornblith}},
  \bibinfo {author} {\bibfnamefont {M.}~\bibnamefont {Norouzi}},\ and\ \bibinfo
  {author} {\bibfnamefont {G.}~\bibnamefont {Hinton}},\ }\bibfield  {title}
  {\bibinfo {title} {A simple framework for contrastive learning of visual
  representations},\ }in\ \href@noop {} {\emph {\bibinfo {booktitle}
  {International conference on machine learning}}}\ (\bibinfo {year} {2020})\
  pp.\ \bibinfo {pages} {1597--1607},\ \bibinfo {note} {tex.organization:
  PMLR}\BibitemShut {NoStop}%
\bibitem [{\citenamefont {Bardes}\ \emph {et~al.}(2022)\citenamefont {Bardes},
  \citenamefont {Ponce},\ and\ \citenamefont {LeCun}}]{bardes2022vicreg}%
  \BibitemOpen
  \bibfield  {author} {\bibinfo {author} {\bibfnamefont {A.}~\bibnamefont
  {Bardes}}, \bibinfo {author} {\bibfnamefont {J.}~\bibnamefont {Ponce}},\ and\
  \bibinfo {author} {\bibfnamefont {Y.}~\bibnamefont {LeCun}},\ }\bibfield
  {title} {\bibinfo {title} {{VICReg}: {Variance}-invariance-covariance
  regularization for self-supervised learning},\ }in\ \href
  {https://openreview.net/forum?id=xm6YD62D1Ub} {\emph {\bibinfo {booktitle}
  {International conference on learning representations}}}\ (\bibinfo {year}
  {2022})\BibitemShut {NoStop}%
\bibitem [{\citenamefont {Zbontar}\ \emph {et~al.}(2021)\citenamefont
  {Zbontar}, \citenamefont {Jing}, \citenamefont {Misra}, \citenamefont
  {LeCun},\ and\ \citenamefont {Deny}}]{zbontar2021barlow}%
  \BibitemOpen
  \bibfield  {author} {\bibinfo {author} {\bibfnamefont {J.}~\bibnamefont
  {Zbontar}}, \bibinfo {author} {\bibfnamefont {L.}~\bibnamefont {Jing}},
  \bibinfo {author} {\bibfnamefont {I.}~\bibnamefont {Misra}}, \bibinfo
  {author} {\bibfnamefont {Y.}~\bibnamefont {LeCun}},\ and\ \bibinfo {author}
  {\bibfnamefont {S.}~\bibnamefont {Deny}},\ }\bibfield  {title} {\bibinfo
  {title} {Barlow twins: {Self}-supervised learning via redundancy reduction},\
  }in\ \href@noop {} {\emph {\bibinfo {booktitle} {International conference on
  machine learning}}}\ (\bibinfo {year} {2021})\ pp.\ \bibinfo {pages}
  {12310--12320},\ \bibinfo {note} {tex.organization: PMLR}\BibitemShut
  {NoStop}%
\bibitem [{\citenamefont {He}\ \emph {et~al.}(2020)\citenamefont {He},
  \citenamefont {Fan}, \citenamefont {Wu}, \citenamefont {Xie},\ and\
  \citenamefont {Girshick}}]{he_momentum_2020}%
  \BibitemOpen
  \bibfield  {author} {\bibinfo {author} {\bibfnamefont {K.}~\bibnamefont
  {He}}, \bibinfo {author} {\bibfnamefont {H.}~\bibnamefont {Fan}}, \bibinfo
  {author} {\bibfnamefont {Y.}~\bibnamefont {Wu}}, \bibinfo {author}
  {\bibfnamefont {S.}~\bibnamefont {Xie}},\ and\ \bibinfo {author}
  {\bibfnamefont {R.}~\bibnamefont {Girshick}},\ }\bibfield  {title} {\bibinfo
  {title} {Momentum {Contrast} for {Unsupervised} {Visual} {Representation}
  {Learning}},\ }in\ \href {https://doi.org/10.1109/CVPR42600.2020.00975}
  {\emph {\bibinfo {booktitle} {2020 {IEEE}/{CVF} {Conference} on {Computer}
  {Vision} and {Pattern} {Recognition} ({CVPR})}}}\ (\bibinfo  {publisher}
  {IEEE},\ \bibinfo {address} {Seattle, WA, USA},\ \bibinfo {year} {2020})\
  pp.\ \bibinfo {pages} {9726--9735}\BibitemShut {NoStop}%
\bibitem [{\citenamefont {Gardner}(1988)}]{gardner_space_1988}%
  \BibitemOpen
  \bibfield  {author} {\bibinfo {author} {\bibfnamefont {E.}~\bibnamefont
  {Gardner}},\ }\bibfield  {title} {\bibinfo {title} {The space of interactions
  in neural network models},\ }\href
  {https://doi.org/10.1088/0305-4470/21/1/030} {\bibfield  {journal} {\bibinfo
  {journal} {Journal of Physics A: Mathematical and General}\ }\textbf
  {\bibinfo {volume} {21}},\ \bibinfo {pages} {257} (\bibinfo {year}
  {1988})}\BibitemShut {NoStop}%
\bibitem [{\citenamefont {Rubin}\ \emph {et~al.}(2010)\citenamefont {Rubin},
  \citenamefont {Monasson},\ and\ \citenamefont
  {Sompolinsky}}]{rubin2010theory}%
  \BibitemOpen
  \bibfield  {author} {\bibinfo {author} {\bibfnamefont {R.}~\bibnamefont
  {Rubin}}, \bibinfo {author} {\bibfnamefont {R.}~\bibnamefont {Monasson}},\
  and\ \bibinfo {author} {\bibfnamefont {H.}~\bibnamefont {Sompolinsky}},\
  }\bibfield  {title} {\bibinfo {title} {Theory of spike timing-based neural
  classifiers},\ }\href@noop {} {\bibfield  {journal} {\bibinfo  {journal}
  {Physical review letters}\ }\textbf {\bibinfo {volume} {105}},\ \bibinfo
  {pages} {218102} (\bibinfo {year} {2010})}\BibitemShut {NoStop}%
\bibitem [{\citenamefont {Schönsberg}\ \emph {et~al.}(2021)\citenamefont
  {Schönsberg}, \citenamefont {Roudi},\ and\ \citenamefont
  {Treves}}]{schonsberg_efficiency_2021}%
  \BibitemOpen
  \bibfield  {author} {\bibinfo {author} {\bibfnamefont {F.}~\bibnamefont
  {Schönsberg}}, \bibinfo {author} {\bibfnamefont {Y.}~\bibnamefont {Roudi}},\
  and\ \bibinfo {author} {\bibfnamefont {A.}~\bibnamefont {Treves}},\
  }\bibfield  {title} {\bibinfo {title} {Efficiency of {Local} {Learning}
  {Rules} in {Threshold}-{Linear} {Associative} {Networks}},\ }\href
  {https://doi.org/10.1103/PhysRevLett.126.018301} {\bibfield  {journal}
  {\bibinfo  {journal} {Physical Review Letters}\ }\textbf {\bibinfo {volume}
  {126}},\ \bibinfo {pages} {018301} (\bibinfo {year} {2021})}\BibitemShut
  {NoStop}%
\bibitem [{\citenamefont {Monasson}(1992)}]{monasson_properties_1992}%
  \BibitemOpen
  \bibfield  {author} {\bibinfo {author} {\bibfnamefont {R.}~\bibnamefont
  {Monasson}},\ }\bibfield  {title} {\bibinfo {title} {Properties of neural
  networks storing spatially correlated patterns},\ }\href
  {https://doi.org/10.1088/0305-4470/25/13/019} {\bibfield  {journal} {\bibinfo
   {journal} {Journal of Physics A: Mathematical and General}\ }\textbf
  {\bibinfo {volume} {25}},\ \bibinfo {pages} {3701} (\bibinfo {year}
  {1992})}\BibitemShut {NoStop}%
\bibitem [{\citenamefont {Lopez}\ \emph {et~al.}(1995)\citenamefont {Lopez},
  \citenamefont {Schroder},\ and\ \citenamefont {Opper}}]{Lopez_1995}%
  \BibitemOpen
  \bibfield  {author} {\bibinfo {author} {\bibfnamefont {B.}~\bibnamefont
  {Lopez}}, \bibinfo {author} {\bibfnamefont {M.}~\bibnamefont {Schroder}},\
  and\ \bibinfo {author} {\bibfnamefont {M.}~\bibnamefont {Opper}},\ }\bibfield
   {title} {\bibinfo {title} {Storage of correlated patterns in a perceptron},\
  }\href {https://doi.org/10.1088/0305-4470/28/16/005} {\bibfield  {journal}
  {\bibinfo  {journal} {Journal of Physics A: Mathematical and General}\
  }\textbf {\bibinfo {volume} {28}},\ \bibinfo {pages} {L447} (\bibinfo {year}
  {1995})}\BibitemShut {NoStop}%
\bibitem [{\citenamefont {Mezard}\ \emph {et~al.}(1986)\citenamefont {Mezard},
  \citenamefont {Parisi},\ and\ \citenamefont {Virasoro}}]{mezard_spin_1986}%
  \BibitemOpen
  \bibfield  {author} {\bibinfo {author} {\bibfnamefont {M.}~\bibnamefont
  {Mezard}}, \bibinfo {author} {\bibfnamefont {G.}~\bibnamefont {Parisi}},\
  and\ \bibinfo {author} {\bibfnamefont {M.}~\bibnamefont {Virasoro}},\ }\href
  {https://doi.org/10.1142/0271} {\emph {\bibinfo {title} {Spin {Glass}
  {Theory} and {Beyond}}}}\ (\bibinfo  {publisher} {WORLD SCIENTIFIC},\
  \bibinfo {year} {1986})\ \bibinfo {note} {\_eprint:
  https://www.worldscientific.com/doi/pdf/10.1142/0271}\BibitemShut {NoStop}%
\bibitem [{\citenamefont {Mezard}\ and\ \citenamefont
  {Montanari}(2009)}]{mezard_information_2009}%
  \BibitemOpen
  \bibfield  {author} {\bibinfo {author} {\bibfnamefont {M.}~\bibnamefont
  {Mezard}}\ and\ \bibinfo {author} {\bibfnamefont {A.}~\bibnamefont
  {Montanari}},\ }\href@noop {} {\emph {\bibinfo {title} {Information, physics,
  and computation}}}\ (\bibinfo  {publisher} {Oxford University Press},\
  \bibinfo {year} {2009})\BibitemShut {NoStop}%
\bibitem [{not()}]{note:SM}%
  \BibitemOpen
  \href@noop {} {\bibinfo {title} {See supplemental material below for details
  of the replica calculations and experimental details, which includes {Ref.}
  [38].}}\BibitemShut {Stop}%
\bibitem [{\citenamefont {Hager}(1989)}]{hager1989updating}%
  \BibitemOpen
  \bibfield  {author} {\bibinfo {author} {\bibfnamefont {W.~W.}\ \bibnamefont
  {Hager}},\ }\bibfield  {title} {\bibinfo {title} {Updating the inverse of a
  matrix},\ }\href@noop {} {\bibfield  {journal} {\bibinfo  {journal} {SIAM
  review}\ }\textbf {\bibinfo {volume} {31}},\ \bibinfo {pages} {221} (\bibinfo
  {year} {1989})}\BibitemShut {NoStop}%
\bibitem [{\citenamefont {Wakhloo}\ \emph {et~al.}(2023)\citenamefont
  {Wakhloo}, \citenamefont {Sussman},\ and\ \citenamefont
  {Chung}}]{note:repository}%
  \BibitemOpen
  \bibfield  {author} {\bibinfo {author} {\bibfnamefont {A.}~\bibnamefont
  {Wakhloo}}, \bibinfo {author} {\bibfnamefont {T.}~\bibnamefont {Sussman}},\
  and\ \bibinfo {author} {\bibfnamefont {S.}~\bibnamefont {Chung}},\
  }\href@noop {} {\bibinfo {title} {Capacity for correlated manifolds code}},\
  \bibinfo {howpublished}
{\url{https://zenodo.org/record/7844169#.ZD9Gwy-B22s}} (\bibinfo {year}
  {2023}),\ \bibinfo {note} {10.5281/zenodo.7844169}\BibitemShut {NoStop}%
\bibitem [{\citenamefont {Chung}\ \emph
  {et~al.}(2018{\natexlab{b}})\citenamefont {Chung}, \citenamefont {Cohen},
  \citenamefont {Sompolinsky},\ and\ \citenamefont
  {Lee}}]{chung_learning_2018}%
  \BibitemOpen
  \bibfield  {author} {\bibinfo {author} {\bibfnamefont {S.}~\bibnamefont
  {Chung}}, \bibinfo {author} {\bibfnamefont {U.}~\bibnamefont {Cohen}},
  \bibinfo {author} {\bibfnamefont {H.}~\bibnamefont {Sompolinsky}},\ and\
  \bibinfo {author} {\bibfnamefont {D.~D.}\ \bibnamefont {Lee}},\ }\bibfield
  {title} {\bibinfo {title} {Learning data manifolds with a cutting plane
  method},\ }\href {https://doi.org/10.1162/neco_a_01119} {\bibfield  {journal}
  {\bibinfo  {journal} {Neural Computation}\ }\textbf {\bibinfo {volume}
  {30}},\ \bibinfo {pages} {2593} (\bibinfo {year}
  {2018}{\natexlab{b}})}\BibitemShut {NoStop}%
\bibitem [{\citenamefont {Boyd}\ \emph {et~al.}(2004)\citenamefont {Boyd},
  \citenamefont {Boyd},\ and\ \citenamefont {Vandenberghe}}]{boyd2004convex}%
  \BibitemOpen
  \bibfield  {author} {\bibinfo {author} {\bibfnamefont {S.}~\bibnamefont
  {Boyd}}, \bibinfo {author} {\bibfnamefont {S.~P.}\ \bibnamefont {Boyd}},\
  and\ \bibinfo {author} {\bibfnamefont {L.}~\bibnamefont {Vandenberghe}},\
  }\href@noop {} {\emph {\bibinfo {title} {Convex optimization}}}\ (\bibinfo
  {publisher} {Cambridge university press},\ \bibinfo {year}
  {2004})\BibitemShut {NoStop}%
\bibitem [{\citenamefont {Gardner}\ and\ \citenamefont
  {Derrida}(1988)}]{gardner1988optimal}%
  \BibitemOpen
  \bibfield  {author} {\bibinfo {author} {\bibfnamefont {E.}~\bibnamefont
  {Gardner}}\ and\ \bibinfo {author} {\bibfnamefont {B.}~\bibnamefont
  {Derrida}},\ }\bibfield  {title} {\bibinfo {title} {Optimal storage
  properties of neural network models},\ }\href@noop {} {\bibfield  {journal}
  {\bibinfo  {journal} {Journal of Physics A: Mathematical and general}\
  }\textbf {\bibinfo {volume} {21}},\ \bibinfo {pages} {271} (\bibinfo {year}
  {1988})}\BibitemShut {NoStop}%
\bibitem [{\citenamefont {He}\ \emph {et~al.}(2016)\citenamefont {He},
  \citenamefont {Zhang}, \citenamefont {Ren},\ and\ \citenamefont
  {Sun}}]{he2016deep}%
  \BibitemOpen
  \bibfield  {author} {\bibinfo {author} {\bibfnamefont {K.}~\bibnamefont
  {He}}, \bibinfo {author} {\bibfnamefont {X.}~\bibnamefont {Zhang}}, \bibinfo
  {author} {\bibfnamefont {S.}~\bibnamefont {Ren}},\ and\ \bibinfo {author}
  {\bibfnamefont {J.}~\bibnamefont {Sun}},\ }\bibfield  {title} {\bibinfo
  {title} {Deep residual learning for image recognition},\ }in\ \href@noop {}
  {\emph {\bibinfo {booktitle} {Proceedings of the {IEEE} conference on
  computer vision and pattern recognition}}}\ (\bibinfo {year} {2016})\ pp.\
  \bibinfo {pages} {770--778}\BibitemShut {NoStop}%
\bibitem [{\citenamefont {Russakovsky}\ \emph {et~al.}(2015)\citenamefont
  {Russakovsky}, \citenamefont {Deng}, \citenamefont {Su}, \citenamefont
  {Krause}, \citenamefont {Satheesh}, \citenamefont {Ma}, \citenamefont
  {Huang}, \citenamefont {Karpathy}, \citenamefont {Khosla}, \citenamefont
  {Bernstein}, \citenamefont {Berg},\ and\ \citenamefont {Fei-Fei}}]{ILSVRC15}%
  \BibitemOpen
  \bibfield  {author} {\bibinfo {author} {\bibfnamefont {O.}~\bibnamefont
  {Russakovsky}}, \bibinfo {author} {\bibfnamefont {J.}~\bibnamefont {Deng}},
  \bibinfo {author} {\bibfnamefont {H.}~\bibnamefont {Su}}, \bibinfo {author}
  {\bibfnamefont {J.}~\bibnamefont {Krause}}, \bibinfo {author} {\bibfnamefont
  {S.}~\bibnamefont {Satheesh}}, \bibinfo {author} {\bibfnamefont
  {S.}~\bibnamefont {Ma}}, \bibinfo {author} {\bibfnamefont {Z.}~\bibnamefont
  {Huang}}, \bibinfo {author} {\bibfnamefont {A.}~\bibnamefont {Karpathy}},
  \bibinfo {author} {\bibfnamefont {A.}~\bibnamefont {Khosla}}, \bibinfo
  {author} {\bibfnamefont {M.}~\bibnamefont {Bernstein}}, \bibinfo {author}
  {\bibfnamefont {A.~C.}\ \bibnamefont {Berg}},\ and\ \bibinfo {author}
  {\bibfnamefont {L.}~\bibnamefont {Fei-Fei}},\ }\bibfield  {title} {\bibinfo
  {title} {{ImageNet} large scale visual recognition challenge},\ }\href
  {https://doi.org/10.1007/s11263-015-0816-y} {\bibfield  {journal} {\bibinfo
  {journal} {International Journal of Computer Vision (IJCV)}\ }\textbf
  {\bibinfo {volume} {115}},\ \bibinfo {pages} {211} (\bibinfo {year}
  {2015})}\BibitemShut {NoStop}%
\bibitem [{\citenamefont {Saxe}\ \emph {et~al.}(2019)\citenamefont {Saxe},
  \citenamefont {McClelland},\ and\ \citenamefont
  {Ganguli}}]{saxe_mathematical_2019}%
  \BibitemOpen
  \bibfield  {author} {\bibinfo {author} {\bibfnamefont {A.~M.}\ \bibnamefont
  {Saxe}}, \bibinfo {author} {\bibfnamefont {J.~L.}\ \bibnamefont
  {McClelland}},\ and\ \bibinfo {author} {\bibfnamefont {S.}~\bibnamefont
  {Ganguli}},\ }\bibfield  {title} {\bibinfo {title} {A mathematical theory of
  semantic development in deep neural networks},\ }\href
  {https://doi.org/10.1073/pnas.1820226116} {\bibfield  {journal} {\bibinfo
  {journal} {Proceedings of the National Academy of Sciences}\ }\textbf
  {\bibinfo {volume} {116}},\ \bibinfo {pages} {11537} (\bibinfo {year}
  {2019})}\BibitemShut {NoStop}%
\bibitem [{\citenamefont {Johnston}\ and\ \citenamefont
  {Fusi}(2023)}]{johnston2023abstract}%
  \BibitemOpen
  \bibfield  {author} {\bibinfo {author} {\bibfnamefont {W.~J.}\ \bibnamefont
  {Johnston}}\ and\ \bibinfo {author} {\bibfnamefont {S.}~\bibnamefont
  {Fusi}},\ }\bibfield  {title} {\bibinfo {title} {Abstract representations
  emerge naturally in neural networks trained to perform multiple tasks},\
  }\href@noop {} {\bibfield  {journal} {\bibinfo  {journal} {Nature
  Communications}\ }\textbf {\bibinfo {volume} {14}},\ \bibinfo {pages} {1040}
  (\bibinfo {year} {2023})}\BibitemShut {NoStop}%
\bibitem [{\citenamefont {Higgins}\ \emph {et~al.}(2017)\citenamefont
  {Higgins}, \citenamefont {Matthey}, \citenamefont {Pal}, \citenamefont
  {Burgess}, \citenamefont {Glorot}, \citenamefont {Botvinick}, \citenamefont
  {Mohamed},\ and\ \citenamefont {Lerchner}}]{higgins2017betavae}%
  \BibitemOpen
  \bibfield  {author} {\bibinfo {author} {\bibfnamefont {I.}~\bibnamefont
  {Higgins}}, \bibinfo {author} {\bibfnamefont {L.}~\bibnamefont {Matthey}},
  \bibinfo {author} {\bibfnamefont {A.}~\bibnamefont {Pal}}, \bibinfo {author}
  {\bibfnamefont {C.}~\bibnamefont {Burgess}}, \bibinfo {author} {\bibfnamefont
  {X.}~\bibnamefont {Glorot}}, \bibinfo {author} {\bibfnamefont
  {M.}~\bibnamefont {Botvinick}}, \bibinfo {author} {\bibfnamefont
  {S.}~\bibnamefont {Mohamed}},\ and\ \bibinfo {author} {\bibfnamefont
  {A.}~\bibnamefont {Lerchner}},\ }\bibfield  {title} {\bibinfo {title}
  {beta-{VAE}: {Learning} basic visual concepts with a constrained variational
  framework},\ }in\ \href {https://openreview.net/forum?id=Sy2fzU9gl} {\emph
  {\bibinfo {booktitle} {International conference on learning
  representations}}}\ (\bibinfo {year} {2017})\BibitemShut {NoStop}%
\bibitem [{\citenamefont {Higgins}\ \emph {et~al.}(2018)\citenamefont
  {Higgins}, \citenamefont {Amos}, \citenamefont {Pfau}, \citenamefont
  {Racaniere}, \citenamefont {Matthey}, \citenamefont {Rezende},\ and\
  \citenamefont {Lerchner}}]{higgins_towards_2018}%
  \BibitemOpen
  \bibfield  {author} {\bibinfo {author} {\bibfnamefont {I.}~\bibnamefont
  {Higgins}}, \bibinfo {author} {\bibfnamefont {D.}~\bibnamefont {Amos}},
  \bibinfo {author} {\bibfnamefont {D.}~\bibnamefont {Pfau}}, \bibinfo {author}
  {\bibfnamefont {S.}~\bibnamefont {Racaniere}}, \bibinfo {author}
  {\bibfnamefont {L.}~\bibnamefont {Matthey}}, \bibinfo {author} {\bibfnamefont
  {D.}~\bibnamefont {Rezende}},\ and\ \bibinfo {author} {\bibfnamefont
  {A.}~\bibnamefont {Lerchner}},\ }\bibfield  {title} {\bibinfo {title}
  {Towards a {Definition} of {Disentangled} {Representations}},\ }\href
  {http://arxiv.org/abs/1812.02230} {\bibfield  {journal} {\bibinfo  {journal}
  {arXiv:1812.02230 [cs, stat]}\ } (\bibinfo {year} {2018})},\ \bibinfo {note}
  {arXiv: 1812.02230}\BibitemShut {NoStop}%
\end{thebibliography}

%

\end{document}


\maketitle
\section{Capacity for General Manifolds with Arbitrary Correlations}
\begin{claim}
Consider a set of manifolds $M^\mu\subset \mathbb{R}^N$ with $\mu= 1,\dots,P$ and with corresponding shape sets $\mathcal S^\mu \subset \mathbb{R}^K$. Suppose the axes and centroids $u^\mu_i \in \mathbb R^N$ are distributed according to: $p(u) \propto \exp\big[- \frac N 2 \sum_{\mu,\nu,i,j,l} (C^{-1})^{\mu,i}_{\nu,j} u^\mu_{i,l} u^\nu_{j,l}\big]$, and assign random binary labels $y^\mu \in \{ -1,1\}$ with equal probability to each manifold. We define the capacity as the maximum number of manifolds per input dimension, $\alpha \equiv P/N$, which admits a solution $w\in\mathbb{S}(\sqrt{N})$ to the separation problem, $\min_{\mu \in \mathbb{N}_1^P} \min_{x\in M^\mu}  y^\mu \langle w, x \rangle \geq \kappa$  with probability $1$ for $\kappa \geq 0$ in the thermodynamic limit, $N, P \to \infty, \;  P/N = O(1).$ Under our assumptions on the manifolds $M^\mu$, the capacity converges to:

\begin{gather}
\frac 1 {\alpha_{cor}(\kappa)} = \frac 1 P \mathbb E_{y} \int D_{y,C}T \min_{V\in \mathcal A} 
||V - T||^2_{y, C},
\label{eq:general-T-supp}
\end{gather}

\noindent where the average is with respect to the i.i.d. labels taking values $\pm 1$ with equal probability, the constraint set is: 

\begin{gather} 
\mathcal A \equiv \bigg\{V \in \mathbb R^{P \times (K+1)} :\forall \mu \in \mathbb N_1^P, \ V^\mu_0 +  \min_{s \in \mathcal{S}^\mu}
\sum_{i>0}V^\mu_i s_i  \geq \kappa 
\bigg\} \; ,
\label{eq:first-constraint}
\end{gather} 

and the Mahalanobis norm $||X||_{y,C}$ is defined by: $||X||^2_{y,C} \equiv \sum_{i,j,\nu,\mu} X^\mu_i X^\nu_j y^\mu y^\nu \big(C^{-1}\big)^{\mu,i}_{\nu, j}$. As in the main text, we also define the Gaussian measure $D_{y,C} T$ as: 

\begin{gather} 
    D_{y,C} T = (2\pi)^{-P(K+1)/2}|G|^{-1/2} \exp\bigg\{ 
-\frac 1 2 \sum_{\mu, \nu, i, j} T^{\mu}_i T^{\nu}_j y^\mu y^\nu \big( C^{-1} \big)^{\mu,i}_{\nu, j}
\bigg\}
\bigg[\prod_{\mu=1}^{P} \prod_{i=0}^K dT_i^\mu \bigg] \;, 
\end{gather} 

where $|G|$ is the determinant of the tensor $y^\mu y^\nu C^{\mu,i}_{\nu,j}$, unrolled into a matrix of dimensions $P(K+1) \times P(K+1).$

\end{claim}

\begin{proof}[Derivation:]
We calculate the log volume of the space of solutions \cite{gardner_space_1988}: 

\begin{gather}
\mathbb E_{y,u} \log Z \equiv \mathbb E_{y,u} \log \int d^Nw \delta(w^2-N) \prod_\mu 
\Theta\bigg( \min_{x \in M^\mu}y^\mu \langle w, x \rangle 
- \kappa 
\bigg)
\label{eq:logvol}
\end{gather} 

This is done using the replica method, which relies on the identity: $\mathbb E \log Z  = \lim_{n\to 0}n^{-1}(\mathbb{E}Z^n - 1) =  \lim_{n\to 0} n^{-1} \log \mathbb E Z^n$ . We first assume that $n \in \mathbb{N}$ and only later take the limit $n\to 0$ after obtaining an expression which is analytic in $n.$ Replicating the volume integral $n$ times and rewriting the constraint in terms of the fields $H^{\mu,a}_i$ gives:

\begin{gather}
\mathbb E_{y,u} Z^n = \mathbb E_{y,u} \int \prod_{a=1}^n d^Nw_a \delta(w_a^2-N) \prod_\mu 
\int \mathbb{D} H^{\mu,a}  \prod_{i=0}^{K}\sqrt{2\pi}\delta(H^{\mu, a}_i - y^\mu w^T_a u^\mu_i),
\end{gather}

where, as in \cite{chung_classification_2018}, we have absorbed the constraint into the measure $\mathbb{D} H:$ 

\begin{gather} 
\mathbb{D} H^\mu = \bigg( \prod_{i=0}^{K} \frac{dH^\mu_i}{\sqrt{2\pi}} \bigg) 
\Theta\bigg(g_{\mathcal S^\mu} (H^\mu) - \kappa\bigg)
\\
g_{\mathcal S^\mu}(H^\mu) = H_0^\mu + \min_{s \in \mathcal{S}^\mu} \sum_{i > 0} H^\mu_i s_i 
\end{gather}

Introducing Fourier representations of the delta functions for $H$ gives:

\begin{gather}
\int \prod_a d^Nw_a \delta(w_a^2-N)\bigg( \prod_\mu 
\int \mathbb{D} H^{\mu,a} \prod_{i=0}^{K}\int \frac{d\hat H^{\mu, a}_{i}}{ \sqrt{2\pi}} \bigg)
\mathbb E_{y,u} \exp\bigg[\sum_{\mu,a,i} i\hat H^{\mu, a}_{i} (H^{\mu, a}_i - y^\mu w^T_a u^\mu_i)\bigg]
\end{gather}

The average over the exponential term is:  

\begin{gather}
\mathbb E_{y,u} \exp\bigg\{
-\sum_{\mu, a, i,l} i\hat H^{\mu, a}_i y^\mu w_a^lu_{i,l}^\mu
\bigg\}
\\=\mathbb E_{y} \int \frac{\big[\prod_{il\mu} du_{il}^\mu\big]}{(2\pi)^{NP(K+1)/2}|C|^{N/2}}
\exp\bigg\{ 
-\frac N 2 \sum_{\mu, i, \nu, j, l } u^\mu_{i,l} u^\nu_{j,l} \big(C^{-1}\big)^{\mu,i}_{\nu,j} - i\sum_{\mu, a, i, l} \hat H^{\mu, a}_iy^\mu w_a^l u_{i,l}^\mu
\bigg\}
\end{gather}

Note that the average over the labels $ y $ cannot be performed analytically. As such, the average $ \mathbb E_y $ will not be written again until the end to avoid clutter. Defining the Cholesky decomposition of $C$, which satisfies $\sum_{\tau, k} L^{\mu,i}_{\tau,k} L^{\nu,j}_{\tau, k} = C^{\mu, i}_{\nu,j},$ we make the change of variables $ u_{i, l}^\mu \mapsto \sum_{\tau, k } L^{\mu, i}_{\tau, k} u_{k, l}^\tau $. This gives:

\begin{gather}
\int \frac{\big[\prod_{il\mu} du_{il}^\mu\big]}{(2\pi)^{NP(K+1)/2}}
\exp\bigg\{ 
-\frac N 2 \sum_{\mu, i, l } \big(u^{\mu}_{i,l}\big)^2 - i\sum_{\mu, a, i, l,\tau,k} \hat H^{\mu, a}_iy^\mu w_a^l L^{\mu,i}_{\tau,k} u_{k,l}^\tau
\bigg\}
\end{gather}

Integrating the $u$ and introducing the overlap matrix $Q_{a,b} = N^{-1}\sum_{l=1}^N w_a^l w_b^l$ then yields: 

\begin{align}
\int &dQ \prod_a \int dw_a \delta(Q_{a,a} -1) \prod_{b=1}^n\delta(N^{-1} w_a^T w_b - Q_{a,b})  \nonumber
\\ &\times
\bigg( \prod_\mu 
\int \mathbb{D} H^{\mu,a} \prod_{i=0}^{K}\int \frac{d\hat H^{\mu, a}_{i}}{\sqrt{2\pi}} \bigg)
\exp\bigg\{ 
-\frac 1 2\sum_{a,b,\nu,\mu, i, j} Q_{a,b} C_{\nu, j}^{\mu, i}\hat H^{\mu, a}_i \hat H^{\nu, b}_jy^\mu y^\nu
+ i\sum_{a,\mu,i}\hat H^{\mu, a}_i H^{\mu, a}_i
\bigg\}
\end{align}

Changing variables:  $ \hat H^{\mu, a }_i \mapsto y^\mu \sum_{\tau, k} [L^{-1}]_{\tau, k}^{\mu, i} \hat H^{\tau, a}_k $ and $ H^{\mu, a}_i \mapsto y^\mu \sum_{\tau, k}L^{\mu, i}_{\tau, k} H^{\tau, a}_k $ gives: 

\begin{align}
\int &dQ \prod_a \int dw_a \delta(Q_{a,a} -1) \prod_{b}\delta(N^{-1} w_a^T w_b - Q_{a,b}) \nonumber
\\
& \times\bigg( \prod_{a} \int_{\mathcal C(y,L)} \bigg[\prod_{\mu,i} dH^{\mu,a}_i\bigg] 
\int \bigg[ \prod_{\mu, i}\frac{d \hat H^{\mu,a}_i}{\sqrt{2\pi}}\bigg] \bigg) 
 \exp\bigg\{ 
-\frac 1 2\sum_{a,b\mu, i} Q_{a,b} \hat H^{\mu, a}_i \hat H^{\mu, b}_i
+ i\sum_{a,\mu,i}\hat H^{\mu, a}_i H^{\mu, a}_i
\bigg\},
\end{align}

where the integral over each matrix $H^{  , a}$ must now be taken over the set: 

\begin{gather}
\mathcal C(y, L) \equiv \bigg\{ 
H \in \mathbb{R}^{P \times (K+1)}: \forall \mu \in \mathbb{N}_1^p, \min_{s^\mu \in \mathcal{S}^\mu} y^\mu \sum_{k,\tau, i} H^{\tau}_k L^{\mu, i}_{\tau, k} s^\mu_i 
\geq \kappa 
\bigg\}  
\label{eq:C-constraint}
\end{gather}

Note that we have defined for convenience the additional element: $s_0^\mu =1,$ and that the sum is taken over: $ 0 \leq k, i \leq K $, and $ 1 \leq \mu, \tau \leq P $. Integrating the $ \hat H $ variables: 

\begin{gather}
\int dQ \prod_a \int dw_a \delta(Q_{a,a} -1) \prod_{b} \delta(N^{-1} w_a^T w_b - Q_{a,b}) \nonumber
\\ \times
\bigg( \prod_a \int_{\mathcal C(y,L)} \bigg[\prod_{\mu,i} dH^{\mu,a}_i\bigg] \bigg)
\exp\bigg\{ 
- \frac 1 2 \sum_{a,b,\mu,i} Q^{-1}_{a,b} H^{\mu, a}_i H^{\mu, b}_i 
- \frac {P(K+1)} 2 \log \det Q 
\bigg\}
\end{gather}

This is almost identical to the formula for the log volume for manifolds with heterogeneous shapes \cite{chung_classification_2018, cohen_separability_2020}. The only difference is that the integration over the $H$ is now constrained to the set $\mathcal C(y,L).$ As such, we can proceed just as in the case of uncorrelated manifolds with different shapes \cite{chung_classification_2018, cohen_separability_2020}, which we briefly describe here. We integrate over the $w_a$ variables by introducing Fourier representations of the delta functions and carrying out the Gaussian integral over the $w$. The auxiliary variables introduced through this process can then be integrated by saddle point, leading to a contribution $\exp\big[ \frac N 2 \log\det Q + \mathrm{const.}\big]$. From here, we assume replica symmetry: $Q_{a,b} = (1-q) \delta_{a,b} + q$ and apply a Hubbard-Stratonovich transformation to the off diagonal term in the summation over $H^{\mu,a}_i H^{\mu,b}_i Q^{-1}_{a,b}$ to arrive at: 

\begin{gather}
    \int \bigg[\prod_{a\neq b} dQ_{a,b}\bigg] \exp\bigg\{\frac{N - P(K+1)} 2 \log\det Q \bigg\} 
    \nonumber
    \\
    \times\int D_I T \bigg[\int_{\mathcal{C}(y,L)} \bigg(\prod_{\mu,i} dH^\mu_i \bigg)
    \exp\bigg\{ 
    -\frac{1}{2(1-q)} ||H - \sqrt{q}T||_2^2 
    + \frac{q}{2(1-q)}||T||^2_2
    \bigg\} 
    \bigg]^n ,
\end{gather}

where just as in the main text, $D_I T$ denotes the isotropic Gaussian measure: $\prod_{\mu,i} dT^\mu_i \exp[-\frac 1 2 (T^\mu_i)^2]/\sqrt{2\pi},$ and we have dropped the constant that emerged from the integral over the $w$. We now apply the identities: $\overline{f(x)^n} \doteq \exp[n \overline{\log f(x)}]$ and $\log\det Q \doteq n\log q + nq/(1-q),$ both of which hold as $n\to 0.$ This gives the expression: 

\begin{gather}
    \int dq \exp\bigg\{\frac{Nn}{2} \bigg( \frac{q}{1-q} + \big(1- \alpha(K+1)\big)\log(1-q)\bigg) + \int D_I T \log \int_{\mathcal{C}(y,L)} 
    \bigg(\prod_{\mu, i}dH^\mu_i\bigg)
    e^{-||H-\sqrt{q}T||^2_2/(2(1-q))}\bigg\} 
\end{gather}

We can see that $q$ will concentrate around its saddle point value in the large $N, P$ limit. This value of $q$ is fixed by the stationarity condition: 

\begin{gather}
    \frac{1}{(1-q)^2} - \frac{1-\alpha(K+1)}{1-q} 
    - \frac 1 {N{(1-q)^2}} \int D_I T \frac
    {\displaystyle\int_{\mathcal{C}(y,L)}\bigg(\prod_{\mu, i}dH^\mu_i\bigg)
    {||H-T||^2_2}e^{-||H-\sqrt{q}T||^2_2 /(2(1-q))}}
    {\displaystyle\int_{\mathcal{C}(y,L)}\bigg(\prod_{\mu, i}dH^\mu_i\bigg)
    e^{-||H-\sqrt{q}T||^2_2 /(2(1-q))}}
    =0
\end{gather}

Placing ourselves in the regime where we are at capacity requires that the volume of solutions shrinks to a single point. In this regime, the overlap between different solutions, $q,$ concentrates about 1 \cite{gardner_space_1988}. Therefore, we can use the above self-consistency condition to determine the value of $\alpha$ such that exactly one solution to the separation problem will exist with probability 1 in the thermodynamic limit. To do this, we note that as $q\to 1$, the dominant terms are those of order $(1-q)^{-2}$, and the integrals over the $H$ variables can be replaced by their values at the saddle point. Thus, to leading order in $N, P$, we can see that: 

\begin{gather}
\frac 1 {\alpha_{corr}(\kappa)}  = \mathbb E_y \int D_I T 
\min_{V \in \mathcal C(y,L) }\frac 1 P \sum_{\mu}^P ||V^\mu - T^\mu||^2_2
\label{eq:constr},
\end{gather} 

where we have reintroduced the expectation over the labels and have switched from using $H$ to $V$ to match the main text. As described below, we use this form of the inverse capacity for all numerical calculations. Changing variables once more: $ V^\tau_k \mapsto  \sum_{\eta, l }[L^{-1}]^{\tau,k}_{\eta, l}y^\eta V_l^\eta $ , and $ T_k^\tau \mapsto \sum_{\eta, l}[L^{-1}]^{\tau, k}_{\eta, l}y^\eta T_l^\eta $ then gives the stated result: 

\begin{gather}
\mathbb E_y \int D_{y,C}T \min_{V \in \mathcal A} \frac 1 P \sum_{\mu , \nu, i, j}
(V_i^\mu - T_i^\mu)(V^\nu_j - T_j^\nu ) \big( C^{-1} \big)^{\mu,i}_{\nu, j} y^\mu y^\nu  \label{eq:1}
\\
\mathcal A = \bigg\{V \in \mathbb R^{P \times (K+1)} :\forall \mu \in \mathbb N_1^P, \ V^\mu_0 +  \min_{s \in \mathcal{S}^\mu}
\sum_{i>0}V^\mu_i s_i  \geq \kappa 
\bigg\}
\label{eq:A-constraint}
\\
D_{y,C} T \equiv (2\pi)^{-P(K+1)/2} |C|^{-1/2} \exp\bigg\{ 
-\frac 1 2 \sum_{\mu, \nu, i, j} T^{\mu}_i T^{\nu}_j y^\mu y^\nu \big( C^{-1} \big)^{\mu,i}_{\nu, j}
\bigg\}
\bigg[\prod_{\mu=1}^{P} \prod_{i=0}^K dT_i^\mu \bigg] \; , 
\end{gather}

where we have used the fact that the determinant of $y^\mu y^\nu C^{\mu,i}_{\nu, j}$ is unaffected by the off-diagonal sign flips $y^\mu y^\nu$ to write the normalizing constant over the $T$ integral as $(2\pi)^{-P(K+1)/2} |C|^{-1/2}$. (This can be derived by considering the matrix integral $\int (\prod_{\mu,i} dT^\mu_i) \exp[-\frac 1 2 y^\mu y^\nu T^{\mu}_i T^\nu_j (C^{-1})^{\mu,i}_{\nu,j}]$ and changing variables $T^\mu_i \mapsto y^\mu T^\mu_i$.) 

\end{proof}

\section{The Capacity for Spheres with Low-Rank Correlations} 

\begin{claim} 
Consider a set of $P$ spheres of radius $r$, intrinsic dimension $K$, and at a distance $r_0$ from the origin, residing in $\mathbb{R}^N$. Given homogenous axis-axis and centroid-centroid correlations as defined in the main text Eq. 6, the capacity for these spheres is given by: 
\begin{gather} 
\frac {1}{\alpha_{corr}(\kappa)} =  K \big(\sqrt q - 1\big)^2  + \int_{-\infty}^{\hat\kappa(q)} \frac{d\xi}{\sqrt{2\pi}} 
e^{
-\frac 1 2\xi^2} 
\big(\xi  - \hat\kappa(q)\big)^2 \; ,
\label{eq:spheres}
\end{gather}

 where the scaled squared norm of the signed fields $ q \equiv \overline{{\sum_{i>0} (V^\mu_i)^2}} / ({(1-\lambda)K})$ and the effective margin $\hat \kappa(q)$ are fixed by the equations: 
 
\begin{gather}
\sqrt{q}  =  1 +  \frac{r\sqrt{1 - \lambda}}{r_0\sqrt{K(1 - \psi)}}
\int_{-\infty}^{\hat\kappa(q)}\frac{d\xi}{\sqrt{2\pi}} 
e^{
-\frac 1 2\xi^2} 
\big(\xi  - \hat\kappa(q) \big)
\\
\hat\kappa(q) = \frac{r\sqrt{K(1-\lambda)q} + \kappa}{r_0\sqrt{1 - \psi}}
\end{gather}

\end{claim} 

\begin{proof}[Derivation:] To start with, we rewrite the formula for the inverse capacity given in Supplementary Eq. 1 in terms of a Gibbs measure \cite{gardner1988optimal}: 

\begin{align}
\frac 1 {\alpha_{cor}(\kappa)} = \lim_{\beta \to \infty }\lim_{P\to\infty}  -\frac 2 P \frac {\partial}{\partial \beta}\mathbb E_{T,y} \log \int &\bigg[ \prod_{\mu, i} dV^\mu_i \bigg] \exp\bigg[ -\frac \beta 2 y^\mu y^\nu \Lambda^{\mu, i}_{\nu, j}(V^\mu_i - T^\mu_i)(V^\nu_j - T^\nu_j)\bigg]
\\
&\times \prod_\mu \Theta\bigg( r_0 V^\mu_0 - \kappa  + \min_{s \in \mathcal S^\mu} 
\sum_{i>0} V^\mu_i s_i r \bigg) \; ,
\nonumber
\end{align}

where $\Lambda^{\mu,i}_{\nu,j}$ is the inverse covariance tensor, $(C^{-1})^{\mu,i}_{\nu,j}$. In this form, we can see that the capacity can be derived from the disorder averaged free energy density, $-(\beta P)^{-1} \mathbb E_{T,y} \log Z $. We calculate this using the replica method \cite{mezard_information_2009, mezard_spin_1986}. Here and in the remainder of the calculation, we implicitly sum over all indices in the exponent unless noted otherwise. That is: $\exp[f(x_{a,b})] \equiv \exp[\sum_{a,b} f(x_{a,b})].$  By the Sherman-Morrison formula \cite{hager1989updating}, $\Lambda^{\mu,i}_{\nu,j}$ has entries: 

\begin{align} 
\Lambda^{\mu,i}_{\nu,j}  &=
\begin{dcases}
\displaystyle
\delta_{ij}\bigg[ \frac {\delta_{\mu,\nu}} {1 - \lambda}  - \frac{\lambda}{(1 - \lambda)(1 + (P-1)\lambda )} \bigg] & \text{for } i > 0 
\\ 
\displaystyle
\frac {\delta_{\mu,\nu}} {1 - \psi}  - \frac{\psi}{(1 - \psi)(1 + (P-1)\psi )} & \text{for } i = j = 0 
\\
\displaystyle
0 & \text{for } i = 0, j \neq 0 
\end{dcases}
\\
&\doteq
\begin{dcases}
\displaystyle
\delta_{ij}\bigg[ \frac {\delta_{\mu,\nu}} {1 - \lambda}  - \frac{1}{P(1 - \lambda)} \bigg] & \text{for } i > 0 
\\ 
\displaystyle
\frac {\delta_{\mu,\nu}} {1 - \psi}  - \frac{1}{P(1 - \psi)} & \text{for } i = j = 0 
\\
\displaystyle
0 & \text{for } i = 0, j \neq 0  \; , 
\end{dcases}
\end{align}

where $\doteq$ denotes equality to leading order in $P$. Using the assumption of spherical manifolds allows us to carry out the constrained minimization in the $\Theta$ functions above: 

\begin{gather} 
r_0 V_0^\mu + \min_{s \in \mathcal S^\mu }\sum_{i>0} r V^\mu_i s_i  =r_0 V_0^\mu - r\sqrt{\sum_{i>0} \big(V^\mu_i\big)^2}  
\label{eq:kkt-constraint}
\end{gather} 

To see this, we apply the KKT conditions to the Lagrangian: $L(s, \eta) = \langle V, s \rangle + \eta (||s||^2 -1)$ \cite{boyd2004convex}. The KKT conditions read:

\begin{align}
    2\eta s  &= - V 
    \\
    \eta &\geq 0 
    \\
    ||s||^2  &\leq 1 
    \\
    \eta (||s||^2 -1 ) & = 0 
\end{align}

For all non-zero V, we must therefore have $\eta > 0,$ from which it follows that $s = - V/||V||$ at the minimum, establishing the identity in Eq. \eqref{eq:kkt-constraint}. Using this simplification and replicating the partition function above $n$ times then gives: 

\begin{gather} 
\mathbb E_{T,y} Z^n = \mathbb E_{T,y} \int \bigg[ \prod_{\mu, i, a} dV^\mu_{a,i}\bigg]
\exp\bigg[ 
 -\frac \beta 2 y^\mu y^\nu \Lambda^{\mu, i}_{\nu, j}(V^\mu_{a,i} - T^\mu_i)(V^\nu_{a,j} - T^\nu_j)
\bigg]  \prod_{a,\mu} \Theta\bigg(V^\mu_{a,0} - \frac{r}{r_0} ||V^\mu_{a,i>0}|| - \frac\kappa r_0 \bigg) \; ,
\end{gather}

where we have used $||V_{i>0}||$ to denote the norm of the last $K$ components of $V^\mu_a$. That is, $||V^\mu_{a,i>0}|| \equiv \sqrt{\sum_{i>0} \big(V^\mu_{a,i}\big)^2}$. The expectation over $ T $ is: 

\begin{gather} 
\int d^{P\times (K+1)}T \exp\bigg[ 
-\frac {1+\beta n} 2 T^\mu_i T^\nu_j y^\mu y^\nu \Lambda^{\mu, i}_{\nu, j}
+ \beta T^\mu_i y^\mu \Lambda^{\mu, i}_{\nu, j} y^\nu V^\nu_{a,j}
- \frac 1 2 \log \det yy^T\circ  \Lambda - \frac{P(K+1)} 2 \log 2\pi 
\bigg]
\\
= 
\exp\bigg[ 
\frac {\beta^2}{2(1+\beta n)}
V^\mu_{a,i}V^\nu_{b,j} y^\mu y^\nu \Lambda^{\mu,i}_{\nu, j}
- \frac {P(K+1)} 2 \log(1 + n\beta)
\bigg]
\end{gather}

Here we have slightly abused notation to write: $(yy^T \circ \Lambda)_{\nu,j}^{\mu,i} = y^\mu y^\nu \Lambda^{\mu,i}_{\nu,j}.$ Note that the replicas $V_a$ are now coupled after integrating out the quenched disorder $T.$ Reinserting this back into the integral and expanding the terms $\log(1+n\beta)$, $\Lambda$, and $1/2(1+\beta n)$ to first order in the $ n\to 0 $ and $ P\to\infty $ limits: 

\begin{align}
\mathbb E_{T,y} Z^n \doteq \mathbb E_y \int &\bigg[ \prod_{\mu, a}\prod_{i>0 } dV^\mu_{a,i}\bigg]
\exp\bigg[ 
 -\sum_{i>0} \frac \beta {2({1 - \lambda)}} \bigg\{ 
  \big(V^\mu_{a,i} \big)^2 
 - \frac 1 {P} \bigg( \sum_\mu y^\mu V^\mu_{a,i}\bigg)^2
 \bigg\} \label{eq:full-int}
 \\
 &+\frac {\beta^2} {2 (1 - \lambda)} \bigg\{ 
  V^\mu_{a,i}V^\mu_{b,i}
 -\frac 1 P \bigg(\sum_{a,\mu}y^\mu V^{\mu}_{a,i}\bigg)^2
 \bigg\}
 - \frac {\beta P(K+1)n} 2 
\bigg] \nonumber
\\
\times\int &\bigg[ \prod_{\mu, a} dV^\mu_{a,0}\bigg]
\exp\bigg[ 
 -\frac \beta {2{(1 - \psi)}} \bigg\{ 
 \big(V^\mu_{a,0} \big)^2 
 -  \frac 1 P \bigg( \sum_\mu y^\mu V^\mu_{a,i}\bigg)^2
 \bigg\} \nonumber
 \\
 &+\frac {\beta^2} {2({1 - \psi})} \bigg\{ 
 V^\mu_{a,i}V^\mu_{b,i}
 -\frac 1 P \bigg(\sum_{a,\mu}y^\mu V^{\mu}_{a,0}\bigg)^2
 \bigg\}
\bigg] 
\prod_{a,\mu} \Theta\bigg(V^\mu_{a,0} - \frac{r}{r_0}||V_{i>0}||- \frac \kappa r_0 \bigg)
\nonumber
\end{align}

Note that we freely ignore constants which do not affect the final result. The important points here are: (1) the only interaction between manifolds happens through the mean, $\sum_\mu y^\mu V^\mu_{a,i}$, and (2) the only interactions between the replicas happens through the quadratic interaction terms, $V^\mu_{a,i} V^\mu_{b,i}$. These considerations motivate the following substitutions, which we enforce using delta functions: 

\begin{gather}
Q^\mu_{a,b} = \frac 1 K \sum_{i>0} V^\mu_{a,i}V^\mu_{b,i}
\\
F_{a,i}= \frac 1 P \sum_{\mu} y^\mu V_{a,i}^\mu 
\end{gather}

In this way, equation \eqref{eq:full-int} becomes: 
\begin{align} 
\mathbb E_y\int
\bigg[\prod_{a \leq b} dQ_{a,b} \bigg] 
\exp\bigg[
&-\frac{K\beta}{2(1-\lambda)}
\big(Q_{a,a}^\mu - \beta Q_{a,b}^\mu \big)
-\frac {P(K+1)\beta n}{2}
+ \log  \mathcal{S}\big(Q \big) 
+ \log \mathcal{U} \big(Q\big)
\bigg],
\end{align}
where we have defined for convenience $Q_{a,b} \equiv Q_{b,a}$, and $\mathcal U $ and $\mathcal S$ are the remaining integrals over the $V^\mu_{a,i}$ and the $F$: 

\begin{align}
\mathcal S(Q) \equiv
 \int \bigg[\prod_{a,i>0} dF_{a,i}\bigg]&\exp
 \bigg[
 \frac {\beta}{2(1 - \lambda)} \bigg\{ 
 P F_{a,i}^2 - \beta P \bigg( \sum_a F_{a,i} \bigg)^2
 \bigg\}
 \bigg]
 \int \bigg[\prod_{a,\mu} \prod_{i>0} dV^\mu_{a,i} \bigg]
 \\
&\times\bigg[\prod_{\mu}\prod_{a\leq b}\delta\bigg(Q_{a,b}^\mu - K^{-1} \sum_{i>0} V^\mu_{a,i} V^\mu_{b,i}\bigg)
\bigg]
\prod_{i,a} \delta \bigg(
F_{a,i} - P^{-1}\sum_\mu  y^\mu V^\mu_{a,i} \bigg)
 \label{eq:S-int}
\end{align} 

\begin{align} 
\mathcal U(Q) \equiv  \int&\bigg[\prod_{a} F_{a,0}\bigg]\exp
 \bigg[
 \frac {\beta}{2(1 - \psi)} \bigg\{ 
 P F_{a,0}^2 - \beta P \bigg( \sum_a F_{a,0} \bigg)^2
 \bigg\}
 \bigg]
 \\
 &\times\int \bigg[\prod_{a,\mu} 
dV^\mu_{a,0}\bigg]
\bigg[ \prod_{i} \delta \bigg(
F_{a,0} - P^{-1}\sum_\mu y^\mu V^\mu_{a,0}
\bigg)\bigg] \nonumber
\prod_{a,\mu} \Theta\bigg(V^\mu_{a,0} - R\sqrt{Q_{a,a} }- \kappa r_0^{-1} \bigg)
\\
&\times\exp\bigg[-\frac \beta {2(1 - \psi)} \big(V^\mu_{a,0} \big)^2 
+ \frac {\beta^2} {2(1 - \psi)} V^\mu_{a,0} V^\mu_{b,0}
\bigg] \label{eq:F-int}
\end{align} 

Note that we have defined $R \equiv \sqrt{K} r/r_0$. 

In order to evaluate $\mathcal U (Q)$ and $\mathcal S(Q)$ in closed form, we now make the replica symmetric ansatz: 

\begin{gather}
Q^{\mu}_{a,b} = \delta_{a,b}(q_0 - q_1) + q_1
\end{gather}

Note that unlike typical assumptions of replica symmetric ansatz \cite{mezard_information_2009, mezard_spin_1986}, we have to assume symmetry across all manifolds: $Q^\mu = Q$. This assumption is motivated by the fact that the interactions between all manifolds are symmetric. Under this assumption, the function $S(Q)$ can be estimated by saddle point as described in ~\cref{lemma:S}. In this way we obtain in the $n\to 0 $ limit: 
\begin{align} 
\int
dq_0 dq_1
\exp\bigg[
&-\frac{nPK\beta}{2(1-\lambda)}
\big(q_0 - \beta (q_0 - q_1) \big)
-\frac {P(K+1)\beta n}{2}
+ \frac 1 2 PK \log\det Q  
+ \log\mathbb E_y \mathcal{U} \big(Q\big)
\bigg],
\end{align}

where in the $n \to 0$ limit we have that: 

\begin{gather} 
\log \det Q = n\log (q_0 - q_1) + n \frac {q_1} {q_0 - q_1}
\end{gather}

In ~\cref{lemma:F}, we evaluate the $\mathbb E_y \mathcal U(Q)$ by saddle point. Plugging this result into the integral gives: 
\begin{align} 
\int
dq_0 dq_1
\exp\bigg[&\frac{PKn}{2}\bigg\{
-\frac{\beta}{1-\lambda}
\big(q_0 - \beta (q_0 - q_1) \big)
-\beta
+ \log (q_0 - q_1) 
+ \frac{q_1}{q_0 - q_1} \bigg\}
\\
&+nP \int \mathcal D_\psi \xi
\log \tail\bigg( \sqrt{\frac{\beta}{1 - \psi}} 
\big( R\sqrt{q_0} + \frac \kappa r_0 - \xi \big)
\bigg)
\bigg]\nonumber,
\end{align} 

where, as below, we have used $\tail$ to denote the unnormalized Gaussian tail function: $\mathcal H (x) \equiv \int_x^\infty dt \exp(-t^2/2)$ and $\mathcal D_\psi \xi$ to denote the zero-mean Gaussian measure with variance $1 - \psi$. Pulling out the $P, n$ factors, we can see that the replicated partition can be written as: 

\begin{gather}
\mathbb E_{T,y} Z^n \doteq \int dq_0 dq_1 e^{nP\mathcal{V}(q_0, q_1)}
\label{last-int} 
\end{gather}

If we estimate the remaining integral by saddle point in the large $P$ limit, we can use the identity $\overline{\log Z} = \lim_{n\to 0} n^{-1} \log \overline{Z^n}$ to arrive at:  

\begin{gather} 
\frac 1 P \mathbb E_{T,y} \log Z \doteq \extr_{q_0, q_1} \mathcal V (q_0, q_1) 
\end{gather}

We therefore have to solve $\nabla V(q_0, q_1) = 0$. Just as in \cref{lemma:F} we can use the expansion of $\tail(\sqrt{\beta} x) $ as $\beta \to\infty$ given in equation \ref{eq:tail-expansion} to see that: 

\begin{align}
\partial_{q_0} \int \mathcal D_\psi \xi
\log \tail\bigg( \sqrt{\frac{\beta}{1 - \psi}} 
\big( R\sqrt{q_0} +  \kappa r_0^{-1} - \xi \big)
\bigg) &=  
- \frac{R\sqrt{\beta}}{{2\sqrt{q_0(1 - \psi)}}}
\int \mathcal D_\psi \xi 
\frac{\exp\bigg[-\frac{\beta(R\sqrt{q_0} + \kappa r_0^{-1} - \xi)^2}{2(1-\psi)}\bigg] }
{\tail\bigg( \sqrt{\frac{\beta}{1 - \psi}} 
\big( R\sqrt{q_0} + \kappa r_0^{-1} - \xi \big)
\bigg)}
\nonumber
\\
&\doteq
- \frac{R{\beta}}{{2(1 - \psi)\sqrt{q_0}}}\int_{-\infty}^{R\sqrt{q_0} + \kappa r_0^{-1}} \mathcal D_\psi \xi (R\sqrt{q_0} - \kappa r_0^{-1} - \xi)
\end{align}

Using this expansion, the function $\mathcal V$ admits two different pairs of saddle points. The meaningful solution is given by:  

\begin{gather} 
q_1 = q_0 - \frac 1 \beta \sqrt{q_0 (1 - \lambda)} 
+ O(\beta^{-2})
\\
\sqrt{q_0}  =  \sqrt{1 - \lambda} +  \frac{r(1 - \lambda)}{r_0\sqrt{K(1 - \psi)}}
\int_{-\infty}^0 \frac{d\xi}{\sqrt{2\pi}} 
e^{
-\frac 1 2\big(\xi + \frac{r\sqrt{Kq_0} +\kappa }{r_0\sqrt{1 - \psi}}\big)^2} 
\xi 
\end{gather}

Denoting the solutions to the above equations as $q_0^*, q_1^*$, the inverse capacity is then given by

\begin{align} 
- 2 \frac{\partial}{\partial \beta}\mathcal V(q_0^*, q_1^*) 
= K \bigg(\sqrt{\frac{q_0^*}{1 - \lambda}} - 1\bigg)^2  + \int_{-\infty}^0  \frac{d\xi}{\sqrt{2\pi}}e^{-\frac 1 2
\big(\xi +\frac{r \sqrt{Kq_0^*} + \kappa }{r_0\sqrt{1 - \psi}} 
\big)^2
} \xi^2,
\end{align}

where again we have used the same expansion of $\tail$ to evaluate the partial derivative with respect to $\beta$. From here, changing $q_0 \mapsto q_0 (1 - \lambda)$ then gives the stated result. 

\end{proof} 

\begin{lemma} \label{lemma:S} 
Under the replica symmetric ansatz and as $\beta,P \to \infty$ and $n\to 0$,  the function $S(Q)$  is asymptotic to: 

\begin{gather}
\exp\bigg[\frac 1 2 PK \log \det Q + \mathrm{const.}  \bigg] \; ,
\end{gather}

where $\mathrm{const.}$ denotes terms which do not depend on $\beta$ or $Q$. 

\end{lemma}

\begin{proof}[Derivation:]
Introducing Fourier representations of the delta functions in equation \ref{eq:S-int} gives: 

\begin{gather} 
\int \bigg[ \prod_{\mu, a\leq b} d \hat Q^\mu_{a,b}\bigg] 
\int \bigg[\prod_{i>0, a} dF_{a,i} d\hat F_{a,i}
\bigg] 
\exp \bigg[ 
\frac i 2 K\sum_{a\leq b} Q_{a,b} \hat Q^\mu_{a,b}
+ \sum_{a,i}\frac {\beta}{2(1 - \lambda)} \bigg\{ 
 P F_{a,i}^2 - \beta P \bigg( \sum_a F_{a,i} \bigg)^2
\bigg\}
\bigg] \nonumber
\\
\times\exp\bigg[  i P F_{a,i} \hat F_{a,i} \bigg]
\int \bigg[\prod_{a,\mu} \prod_{i>0} dV^\mu_{a,i}
\bigg]
\exp\bigg[ 
-\frac i 2 \sum_{a\leq b}\hat Q^\mu_{a,b} V^\mu_{a,i} V^\mu_{b,i}
- i y^\mu \hat F_{a,i} V^\mu_{a,i}
\bigg] 
\end{gather}

It is convenient to now rewrite the ordered sum over $a\leq b$ indices as an unordered sum over all pairings $(a,b).$ As above, we define $Q_{b,a} \equiv Q_{a,b}$ and $\hat Q_{b,a} \equiv Q_{a,b}$ for $ a > b$. If we further change $\hat Q_{a,b} \mapsto 2\hat Q_{a,b}$ for all $a\neq b$, we can eliminate the unordered sums as desired, leaving: 

\begin{gather} 
\int \bigg[ \prod_{\mu, a\leq b} d \hat Q^\mu_{a,b}\bigg] 
\int \bigg[\prod_{i>0, a} dF_{a,i} d\hat F_{a,i}
\bigg] 
\exp \bigg[ 
\frac i 2 K Q_{a,b} \hat Q^\mu_{a,b}
+\frac {\beta}{2(1 - \lambda)} \bigg\{ 
 P F_{a,i}^2 - \beta P \bigg( \sum_a F_{a,i} \bigg)^2
\bigg\}
\bigg] \nonumber
\\
\times\exp\bigg[  i P F_{a,i} \hat F_{a,i} \bigg]
\int \bigg[\prod_{a,\mu} \prod_{i>0} dV^\mu_{a,i}
\bigg]
\exp\bigg[ 
-\frac i 2 \hat Q^\mu_{a,b} V^\mu_{a,i} V^\mu_{b,i}
- i y^\mu \hat F_{a,i} V^\mu_{a,i}
\bigg] \; ,
\end{gather}

where as usual we neglect constants which do not depend on either the variables of integration or $\beta$, as they will not affect the final result. The integrals over the $V$ followed by the $\hat F$ variables are now both standard Gaussian integrals. They yield: 

\begin{gather}
\int \bigg[ \prod_{\mu, a\leq b} d \hat Q^\mu_{a,b}\bigg] \bigg[\prod_{a}\prod_{i>0} dF_{a,i}\bigg] 
\exp\bigg[ 
\frac i 2 K Q_{a,b} \hat Q^\mu_{a,b} 
- \frac K 2 \sum_\mu \log \det \hat Q^\mu 
- \frac 1 2 \log \det \bigg( \sum_\mu \big(\hat Q^\mu\big)^{-1} \bigg)
 \nonumber
\\
-\frac {P^2} 2 F_{a,i} F_{b,i}\bigg( \sum_\mu (\hat Q^\mu )^{-1}\bigg)^{-1}_{a,b}
+\frac {\beta}{2(1 - \lambda)} \bigg\{ 
 P F_{a,i}^2 - \beta P \bigg( \sum_a F_{a,i} \bigg)^2
\bigg\}
\bigg]
\end{gather}

The remaining Gaussian integral over the $F$ produces a term which is subleading in $P$. Therefore, the leading order in the remaining integral over the $\hat Q$ is simply the first two terms in the exponent. If we now invoke the replica symmetric ansatz on the conjugate variables: $\hat Q^\mu = \hat Q$, we are left with:  

\begin{gather} 
\int \bigg[\prod_{a \leq b} \hat Q_{a,b}\bigg] 
\exp\bigg[ 
\frac i 2 PK Q_{a,b} \hat Q_{a,b}
- \frac {PK} 2 \log \det \hat Q
\bigg] 
\end{gather}

Estimating this integral by saddle point then gives the desired result. 
\end{proof}

\begin{lemma}\label{lemma:F}
Under the replica symmetric ansatz, as $\beta,P \to \infty$ and $n\to 0$, the function $\mathcal U(Q)$ is asymptotic to:

\begin{gather} 
\mathcal U(Q) \doteq \exp\bigg[nP \int \mathcal{D}_{\psi} \xi
\log \tail\bigg( \sqrt{\frac{\beta}{1 - \psi}} 
\big( R\sqrt{q_0} + \frac \kappa {r_0} - \xi \big)
\bigg)
+ \frac 1 2\beta P n
+ \mathrm{const.} 
\bigg] 
\; , 
\end{gather}

where, as above, $\mathrm{const.}$ denotes terms which do not depend on $\beta$ or $Q$, and $\mathcal H$ is the unnormalized Gaussian tail function: 

\begin{gather}
\mathcal H(x) \equiv \int_x^\infty ds e^{-s^2/2}
\end{gather}

We also use $\mathcal D_\psi \xi$ to denote the Gaussian measure:

\begin{gather}
 \mathcal D_\psi \xi \equiv   \frac{d\xi e^{-\xi^2/2(1 -\psi)} }{\sqrt{2\pi (1 - \psi)}} 
\end{gather}

\end{lemma}

\begin{proof}[Derivation:]
Introducing Fourier representations of the delta functions in equation \eqref{eq:F-int}:
\begin{align}
\mathbb E_y \int &\bigg[ \prod_a  d F_{a,0}d\hat F_a \prod_\mu dV^\mu_{a,0}\bigg] \exp\Bigg[
- \frac \beta {2(1 - \psi)} \big(V^\mu_{a,0}\big)^2
+ \frac {\beta^2}{2(1 - \psi)} \bigg(\sum_a V^\mu_{a,0}\bigg)^2
\bigg]
\prod_{a,\mu}\Theta\bigg(V^\mu_{a,0} - R\sqrt{q_0} - \frac \kappa {r_0} \bigg)
\nonumber
\\
&\times
\exp\bigg[\frac{P\beta}{2(1 -\psi)}
\bigg\{F_{a,i}^2
-\bigg(\sum_a F_{a,i}\bigg)^2\bigg\} 
-\frac{i\beta}{1 - \psi}\hat F_a\big(P F_{a,0}  - y^\mu V_{a,0}^\mu\big)
+ O(\log\beta)
\bigg]
\end{align} 

Note that the terms $O(\log\beta)$ will have no effect on our final answer, so we ignore them. If we now make the additional replica symmetric assumption: $\hat F_a =\hat F, F_a  = F$ and introduce a Hubbard-Stratonovich transform on each of the terms $(\sum_a V^\mu_{a,0})^2$, we have: 

\begin{align}
\int dF d\hat F
&\exp\bigg[
\frac{P\beta n}{2(1 - \psi)} F^2 
- \frac{iPn\beta}{1 - \psi} F \hat F+ O(n^2)
\bigg] 
\prod_\mu\mathbb E_{y^\mu} \int \mathcal D_\psi \xi_\mu
\\
&\times \bigg[\int dV^\mu_0 \exp\bigg[ 
-\frac{\beta}{2(1 - \psi)} \big(V^\mu_0\big)^2 + \frac{\beta}{1 - \psi} V^\mu_0 (\xi_\mu + i \hat F y^\mu) 
\bigg] \Theta\bigg(V^\mu_0 - R\sqrt{q_0} - \frac \kappa {r_0}\bigg)\bigg]^n 
\end{align}

We can factorize the integral across the $\mu$-index and absorb the $\Theta$ function into the limits of integration of the $V_0$ variables to obtain: 

\begin{align} 
\int d F d\hat F 
&\bigg[
\int \mathcal D_\psi \xi
\mathbb E_y
\bigg(
\int_{R\sqrt{q_0} + \kappa}^\infty
d V_0
\exp
\bigg[ 
- \frac \beta {2(1 - \psi)}
\big(V_{0} - \xi- i\hat F y)^2
+ \frac \beta {2(1 - \psi)} (\xi + i \hat F y)^2
\bigg]
\bigg)^n\bigg]^P
\nonumber
\\
&\times\exp\bigg[-\frac{i Pn\beta F \hat F}{1 - \psi}
+\frac{P\beta n}{2(1 -\psi)}
F^2
\bigg]
\end{align}

Using the identity $\mathbb E_x f(x)^n \doteq e^{n \mathbb E_x\log f(x)} $, which is valid in the $n\to 0$ limit, we obtain, after changing variables $V \mapsto V + \xi + i \hat F y$:

\begin{align} 
\int dF d \hat F &\exp\bigg[ 
nP \mathbb E_y \int
\mathcal D_\psi \xi \log \bigg( 
\int_{R\sqrt{q} - \xi - i\hat F y}^\infty e^{-\beta V_0^2 /2(1-\psi)}
\bigg)
+ nP \mathbb E_y 
\int \mathcal D_\psi \xi 
\frac{\beta}{2(1 - \psi)} (\xi+i\hat F y)^2
\bigg]
\\
&\times\exp\bigg[-\frac{i Pn\beta F \hat F}{1 - \psi}
+\frac{P\beta n}{2(1 -\psi)}
F^2
\bigg]
\end{align} 

Carrying out the expectations over the terms $(\xi + i\hat F y)^2 $ and changing $V\mapsto \sqrt{(1-\psi)/\beta} V$ then gives: 

\begin{align} 
\int dF d \hat F
\exp\bigg[&nP\mathbb E_y \int\mathcal D_\psi \xi
\log \tail\bigg( \sqrt{\frac{\beta}{1 - \psi}} 
\big( R\sqrt{q_0} + \kappa r_0^{-1} - \xi - i\hat F y \big)
\bigg) \nonumber
\\
&+ \frac 1 2 \beta P n 
- \frac {nP\beta}{2(1 - \psi)} \hat F^2
-\frac{i Pn\beta F \hat F}{1 - \psi}
+\frac{P\beta n}{2(1 -\psi)}
F^2
+ O(nP\log \beta) 
\bigg]  \label{eq:f0-int} 
\end{align}

We are now ready to estimate this integral by saddle point. To do so, we start by noting that the partial derivative with respect to $\hat F$ of the average over the $\log \tail$ term is:

\begin{gather} 
\sqrt{\frac{\beta}{1 - \psi}}\mathbb E_{y}\int \mathcal D_\psi \xi
\frac{
yi
\exp\bigg[
-\frac{\beta}{2(1 -\psi)}
\big(R\sqrt{q_0} + \kappa r_0^{-1} - \xi - i\hat F y\big)^2\bigg]}
{
\tail\bigg[
\sqrt{\frac{\beta}{1 - \psi}} 
\big( R\sqrt{q_0} + \kappa r_0^{-1} - \xi - i\hat F y^\mu \big)
\bigg]
}
\end{gather} 

In the $\beta\to\infty$ limit, we can use the expansion: 

\begin{gather} 
\tail(\sqrt \beta x) \doteq \begin{cases}
\frac{e^{-\beta x^2/2}}{\sqrt \beta x} & x >0 
\\
\sqrt{2\pi} & x < 0 
\end{cases}
\label{eq:tail-expansion} 
\end{gather} 

Invoking this expansion, the expression simplifies to: 

\begin{gather} 
\frac{i\beta}{1 - \psi}\mathbb E_{y}
\int_{-\infty}^{0}\frac{d\xi e^{-(\xi+R\sqrt{q_0}
+ \kappa  r_0^{-1}
- i \hat F y)^2/2(1 -\psi)} }{\sqrt{2\pi (1 - \psi)}}\xi y
\end{gather}

The saddle points over the $F, \hat F$ variables then satisfy the self-consistent equations:

\begin{gather} 
F= i\hat F 
\\
- i F -\hat F + i\mathbb E_{y}
\int_{-\infty}^{0}\frac{d\xi e^{-(\xi+R\sqrt{q_0}
+ \kappa  r_0^{-1}
- i \hat F y)^2/2(1 -\psi)} }{\sqrt{2\pi (1 - \psi)}}\xi y =0,
\end{gather}

which have the solution $\hat F = F = 0$. Replacing equation \eqref{eq:f0-int} with its value at the saddle point, we can see that the function $\mathbb E_y U(Q)$ can be stated as: 

\begin{gather}
\exp\bigg[nP\int \mathcal D_\psi \xi
\log \tail\bigg( \sqrt{\frac{\beta}{1 - \psi}} 
\big( R\sqrt{q_0} + \kappa r_0^{-1} - \xi \big)
\bigg)
+ \frac 1 2 \beta Pn 
\bigg],
\end{gather}

which is what we wanted to show.
\end{proof}
\pagebreak 

\begin{figure*}
\begin{algorithm}[H]
\caption{Capacity Estimation for Correlated Data ($\alpha_{cor}$)} \label{alg:get-cap}
\begin{flushleft}
    \textbf{Input:} $n_t:$ Number of Monte Carlo draws. $G:$ Data array of shape $P \times N \times M$ containing $M$ samples of $P$ distinct manifolds, in an ambient dimension of $N$.
    \\
    \textbf{Output:} $\alpha:$ Capacity estimate.
\end{flushleft} 
\begin{algorithmic}[1]
    \State{$L, S \gets \text{GetShapesAndCholesky}(G)$}
    \State{$\alpha^{-1} \gets \text{ZerosArray}(n_t)$}
    \For{$i$ from $1$ to $n_t$}
        \State{$y \gets Bernoulli^{\otimes P}(0.5)$}
        \State{$T \gets Normal^{\otimes P\times M+1}(0,1)$}
        \State{$\alpha^{-1}[i] \gets \min_{V \in \mathcal C(y, L) }\frac 1 P \sum_{\mu}^P ||V^\mu - T^\mu||^2 $}
    \EndFor
    \State \textbf{return} $1/\text{mean}(\alpha^{-1})$
\end{algorithmic} 
\end{algorithm} 
\end{figure*}

\begin{figure*}
\begin{algorithm}[H]
\label{alg:getshapes}
    \caption{GetShapesAndCholesky}
    \begin{flushleft} 
        \textbf{Input:} $G:$ Data array of shape $P \times N \times M$ containing $M$ samples of $P$ distinct manifolds, in an ambient dimension of $N$.\\
        \textbf{Output:} $L:$ Cholesky decomposition of the correlation tensor, reshaped into a $P (M +1) \times P (M+1)$ matrix. $S:$ An array of shape $P \times M \times M$ containing data points in their axis coordinates.
    \end{flushleft}
    \begin{algorithmic}[1]
    \State{$Ax, S \gets \text{ZerosArray}(P, N, M+1), \text{ZerosArray}(P, M, M)$}
    \For{$\mu$ from $1$ to $P$}
        \State{$c \gets M^{-1}\sum_{i=1}^{M} G[\mu, :, i]$\Comment{Get this manifold's centroid}}
        \For{$i$ from $1$ to $M$}
            \State{$G[\mu, :, i] \gets G[\mu, :, i] - c$ \Comment{Center each sample}} 
        \EndFor
        \State{$Ax[\mu,:, 0] \gets c$ \Comment{Assign centroid to leading dimension}}
        \State{$Ax[\mu, :, 1:], S[\mu] \gets \mathrm{QR}(G[\mu])$\Comment{Assign axes to the next $m$ dimensions, and get the manifold points in the manifold axis coordinates, $\mathcal{S}^\mu$}}
    \EndFor
    \State{$Ax \gets \mathrm{reshape}(Ax, (P*(M+1), N))$}
    \If{$\mathrm{Rank}(Ax) < P*(M+1)$}
        \State{$C \gets Ax Ax^T + \epsilon I$ \Comment{Perturb the diagonal by a small $\epsilon$ to make $C$ positive definite (we set $\epsilon =0.001$).}}
        \State{$L \gets \mathrm{Cholesky}(C)$}
    \Else
        \State{$Q, L^T \gets \mathrm{QR}(Ax)$\Comment{When full rank, use the QR decomposition to get $L$.}}
    \EndIf
    \State \textbf{return} $L, S$
\end{algorithmic} 
\end{algorithm} 
\end{figure*}

\clearpage

\section{Numerical Implementation of the Capacity Estimator:}

In this section, we describe how we estimate the capacity for arbitrary data manifolds. While there are several ways to parameterize the data manifolds in terms of axes and shape sets, we use the QR decomposition of the matrices containing the (centered) manifold data points to obtain a set of orthogonal axes vectors (Q), together with the coordinates of each manifold point in this orthogonal basis (R), as described in steps 1-8 of Algorithm 2. The shape sets $\mathcal S^\mu$ are then simply taken to be the manifold points in this basis. With respect to the quadratic minimization in Eqs \eqref{eq:general-T-supp} and \eqref{eq:numrcl-T} with linear constraint sets enforcing separability in Eqs. \eqref{eq:first-constraint} and \eqref{eq:C-constraint-2}, this choice is equivalent to taking the shape sets to be the convex hull of all manifold points \cite{boyd2004convex}. With these definitions in hand, the correlation tensor $C$ can then simply be taken to be the empirical covariance tensor, $C^{\mu,i}_{\nu,j}=\langle u^\mu_i, u^\nu_j \rangle.$ 

\hspace{1cm} The direct estimation of Supplementary Eq. \eqref{eq:general-T-supp} using the data manifolds parameterized as described above is a difficult problem. The main difficulty comes from the inversion of large correlation tensors, which is a highly numerically unstable operation. We sidestep this difficulty by changing variables: $(V - T)^\mu_i \mapsto \sum_{\tau,k} y^\mu L^{\mu, i}_{\tau, k} (V-T)^\tau_k$. Here $L$ is the Cholesky factorization of the correlation tensor, $C$, which satisfies $\sum_{\tau,k} L^{\mu,i}_{\tau, k} L^{\nu,j}_{\tau,k} = C^{\mu,i}_{\nu,j}$. This change of variables gives the representation: 

\begin{gather}
    \frac 1{\alpha_{cor}(\kappa)} = \overline{\int D_I T 
    \min_{V \in \mathcal C(y, L) }\frac 1 P \sum_{\mu}^P ||V^\mu - T^\mu||^2 }
    \label{eq:numrcl-T}
    \\
    \mathcal C(y, L) = \bigg\{ 
    V \in \mathbb{R}^{P\times(K+1)}:
    \min_{s^\mu \in \mathcal{S}^\mu} y^\mu\sum_{k,\tau, i}  V^{\tau}_k L^{\mu, i}_{\tau, k} s^\mu_i
    \geq \kappa 
    \bigg\} \; , 
    \label{eq:C-constraint-2}
\end{gather}

\noindent where we have defined $s^\mu_0 \equiv 1$ for convenience. We can see that estimating Eq. \eqref{eq:numrcl-T} with Monte Carlo draws of $y, T$ now only requires calculating the Cholesky factorization of the covariance. Even when $N$ is very large, this step can safely be done using the QR decomposition of the matrix containing manifold axes and centroids (Algorithm 2, line 15). Note, however, that this step requires that the correlation tensor be full rank, which may not always be the case (e.g., when $N < MP$). Therefore, when the covariance tensor is not full rank, we add a small perturbation to the diagonal of the correlation tensor and calculate the Cholesky factor directly (Algorithm 2, lines 11-14; see also \cite{note:repository}). 

\hspace{1cm} Once we have the Cholesky factorization of the covariance, we minimize the integrand of \eqref{eq:numrcl-T} using standard convex optimization routines. In this way, we can accurately estimate the capacity for arbitrary data manifolds by following the pseudocode in Algorithms (1) and (2).


\section{Gaussian Point Cloud Simulation}

 For these simulations, we used $P=80$ point cloud manifolds, each made up of the convex hulls of $40$ random vectors in $\mathbb R^{3800}$, and we averaged results over 5 runs. These vectors were the sum of two Gaussian vectors: a manifold centroid $u^\mu_0$, and a sample-specific vector, $x_j^\mu$. Each manifold was then defined as $M^\mu = \mathrm{conv}\{u_0^\mu + x^\mu_j : j \in \mathbb N_1^{40}\}$. The sample vectors and centroids were each drawn from zero-mean multivariate Gaussian distributions with block covariance matrices with the strength of the off-diagonal terms of both matrices being uniformly scaled from trial to trial (see Supplementary Fig. \ref{fig:sm-intensity} above for an example). That is, given an intensity $\gamma$, we enforced $\overline{\langle u_0^\mu, u_0^\nu \rangle} = \gamma C^{\mu,\nu}_{cent}$ for $\mu \neq \nu$, while each of the sample vectors satisfied $\overline{\langle x^\mu_j, x^\nu_i \rangle} = \gamma C^{\mu,\nu}_{samp}$ for $\mu\neq \nu$. The average norms, $||u^\mu_0||^2$ and $||x^\mu_i||^2$ (i.e., the diagonal elements of $C$) were respectively fixed to $25$ and $1.$ 
 
 \begin{figure} 
 \includegraphics[width=1.
 \textwidth]{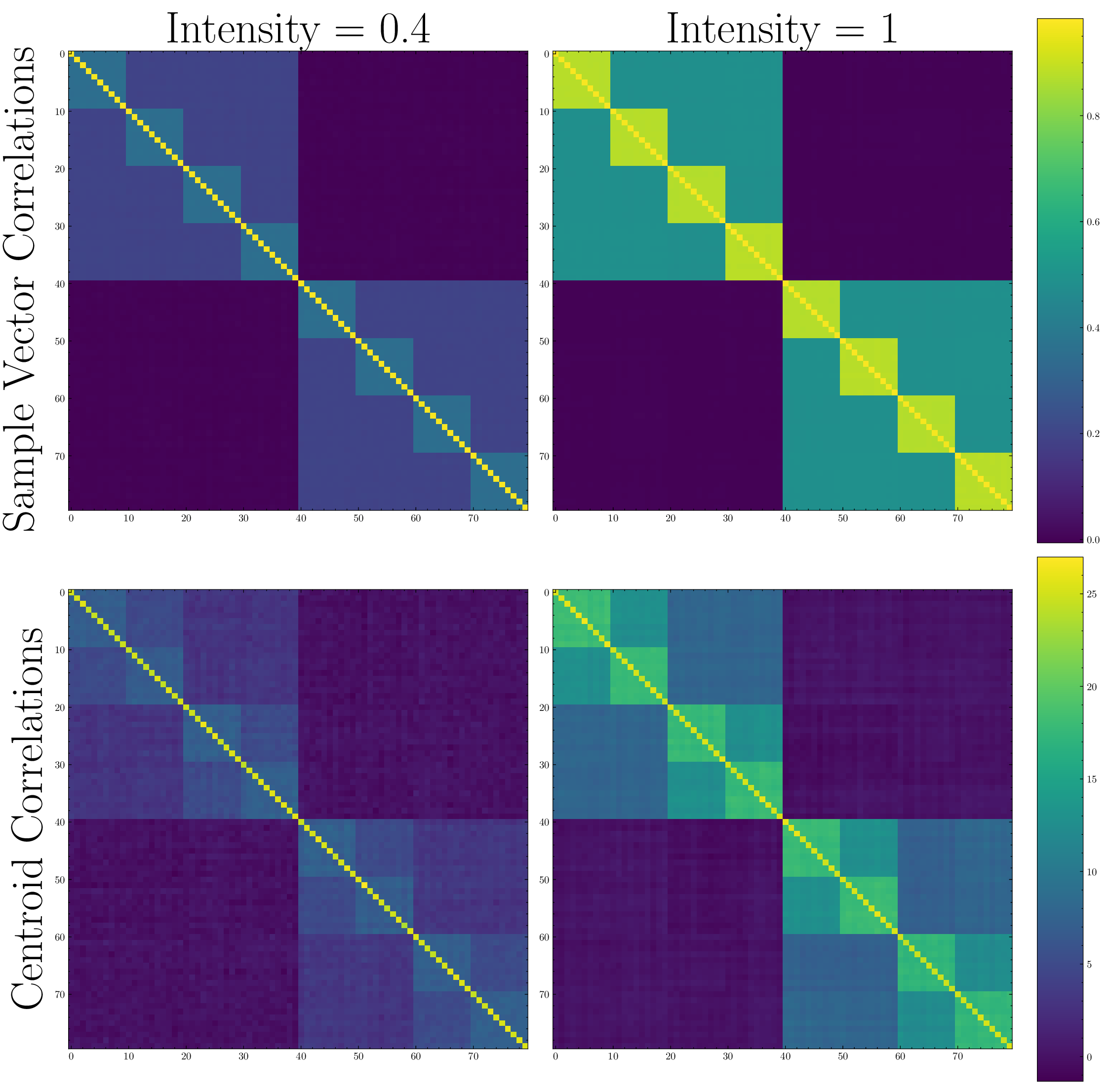}
 \caption{\label{fig:sm-intensity}
 Example of the block covariance matrices for the sample-specific vectors (top row) and centroids (bottom row) at low intensity (left column) and high intensity (right column).
 }
 \end{figure}

 \clearpage
 
\section{ResNet50 Analyses} 

We generated object manifolds from the ResNet 50 architecture trained with SimCLR \cite{chen2020simple, he2016deep} using a similar procedure as in \cite{cohen_separability_2020}. On each of the five experimental runs, we first randomly selected $P=70$ ImageNet classes and chose $45$ samples per class. We then extracted the activations from a sub-sample of 13 out of the 49 ReLU layers in the network, as well as the final average pooling operation. As described in \cite{he2016deep}, these ReLU layers are the result of applying the rectified linear non-linearity $\mathrm{ReLU}(x) = \max\{0,x\}$ elementwise to the outputs of convolutional layers. The layer width of the final average pooling layer was 2,048, while the layer widths of the ReLU layers ranged from 25,088 to 802,816. Given the size of these layers we projected the activations of the raw input and the ReLU layers onto $N=8,000$ vectors sampled randomly from the unit sphere in order to conserve memory as in \cite{cohen_separability_2020}. We then applied the low rank approximation from \cite{cohen_separability_2020}, the simulation capacity algorithm from \cite{chung_learning_2018}, and our $\alpha_{cor}$ estimator to the projected manifolds. Note that we perform the random projection before performing any of the steps described in Supplementary Section 3. Code reproducing all analyses may be found in \cite{note:repository}.
